\documentclass[pdftex,envcountsect,envcountsame]{llncs}

\usepackage{xspace}
\usepackage{textgreek}
\usepackage{latexsym,amssymb,amsmath,amsfonts}
\usepackage{paralist}
\usepackage{todonotes}
\usepackage{tabularx}
\usepackage{enumerate}

\usepackage[latin1]{inputenc}

\usepackage[pdftitle   = {{Combining Cover Algorithms}},
            pdfauthor  = {{Calvanese - Ghilardi - Gianola - Montali - Rivkin}},
            pdfsubject = {{}}]{hyperref}

\usepackage{url}
\usepackage{comment}
\usepackage{graphicx}
\usepackage{tikz-cd}
\usepackage{wrapfig}

\usepackage{amsmath,amssymb}
\usepackage{mathptmx}
\DeclareMathAlphabet{\mathcal}{OMS}{cmsy}{m}{n}

\usepackage[linesnumberedhidden,ruled,vlined]{algorithm2e}
\usepackage{graphicx}

\title{Combined Covers and Beth Definability}

\author{Diego Calvanese$^1$, Silvio Ghilardi$^2$,  Alessandro Gianola$^{1,3}$, \\
 Marco Montali$^1$, Andrey Rivkin$^1$}

\institute{%
 $^1$Faculty of Computer Science, Free University of Bozen-Bolzano, Bolzano, Italy\\
\{calvanese, gianola, montali, rivkin\}@inf.unibz.it \\
 $^2$Dipartimento di Matematica, Universit\`a degli Studi di Milano, Milan, Italy\\
 silvio.ghilardi@unimi.it\\
 $^3$CSE Department, University of California San Diego, San Diego (CA), USA\\
agianola@eng.ucsd.edu
}

\newcommand{\tup}[1]{\langle #1\rangle}            

\newcolumntype{C}{>{\centering\arraybackslash}X}

\makeatletter
\g@addto@macro\normalsize{%
\setlength{\abovecaptionskip}{-2pt}
\setlength\abovedisplayskip{3pt}
\setlength\belowdisplayskip{3pt}
\setlength\abovedisplayshortskip{3pt}
\setlength\belowdisplayshortskip{3pt}
}
\makeatother

\newcounter{dummy} 
\newcounter{dummy1} 
\newcounter{dummy2}
\newcounter{dummy3} 
\newcounter{dummy4}
\newcounter{dummy5} 
\newcounter{dummy6}

\usepackage[thmmarks,amsmath]{ntheorem}
\theorempreskip{1pt}
\theorempostskip{1pt}

\theoremseparator{.}
\theorembodyfont{\itshape}
\theoremsymbol{$\triangleleft$}
\newtheorem{theorem}[dummy]{Theorem}
\newtheorem{lemma}[dummy1]{Lemma}
\newtheorem{definition}[dummy2]{Definition}
\newtheorem{proposition}[dummy3]{Proposition}

\theorembodyfont{\normalfont}
\newtheorem{example}[dummy4]{Example}
\newtheorem{remark}[dummy5]{Remark}

\theoremstyle{nonumberplain}
\theoremheaderfont{\itshape}
\theorembodyfont{\normalfont}

\theoremseparator{.}
\theoremsymbol{\ensuremath{\dashv}}
\newtheorem{proof}[dummy6]{Proof}

\qedsymbol{\ensuremath{\dashv}}

\newcommand{\ua}{\ensuremath{\underline a}}
\newcommand{\ub}{\ensuremath{\underline b}}
\newcommand{\uc}{\ensuremath{\underline c}}

\newcommand{\ue}{\ensuremath{\underline e}}

\newcommand{\ut}{\ensuremath{\underline t}}

\newcommand{\ux}{\ensuremath{\underline x}}
\newcommand{\uy}{\ensuremath{\underline y}}
\newcommand{\uz}{\ensuremath{\underline z}}

\newcommand{\um}{\ensuremath{\underline m}}

\newcommand{\ueta}{\ensuremath{\underline \eta}}
\newcommand{\uxi}{\ensuremath{\underline \xi}}

\newcommand{\cA}{\ensuremath \mathcal A}

\newcommand{\cM}{\ensuremath \mathcal M}
\newcommand{\cN}{\ensuremath \mathcal N}
\newcommand{\cS}{\ensuremath \mathcal S}

\newcommand{\cI}{\ensuremath \mathcal I}
\newcommand{\lra}{\longrightarrow}

\renewcommand{\int}{\ensuremath {\mathcal I}}

\newcommand{\EUF}{\ensuremath{\mathcal{EUF}}}
\newcommand{\IDL}{\ensuremath{\mathcal{IDL}}}

\definecolor{deepblue}{HTML}{0C3B80}
\definecolor{deepgreen}{HTML}{2EA601}
\definecolor{lightOrange}{HTML}{FFA03C}
\definecolor{darkOrange}{HTML}{F1800A}
\definecolor{lightBlue}{HTML}{0174CD}
\definecolor{greenF}{HTML}{2CBB5C}
\definecolor{cyan}{HTML}{86A6D5}

\tikzstyle{sortnode} = [
  rectangle,
  rounded corners=5pt,
  minimum width=18mm,
  minimum height=6mm,
  very thick,
]

\tikzstyle{functnode} = [
  rectangle,
  minimum width=1cm,
  minimum height=5mm,
]

\tikzstyle{idnode} = [
  sortnode,
  draw=deepblue,
  fill=cyan!50,
  text=deepblue,  
]

\tikzstyle{valnode} = [
  sortnode,
  densely dashed,
  draw=violet,
  fill=magenta!50,
  text=violet,
]

\tikzstyle{f} = [
  thick,
]

\tikzstyle{fd} = [
  f,
  -angle 90,
]

\tikzstyle{relation}=[rectangle split, rectangle split parts=#1, rectangle split part align=base, draw, anchor=center, align=center, text height=3mm, font=\bfseries, text centered]

\tikzstyle{sortnode} = [
  rectangle,
  rounded corners=5pt,
  minimum width=18mm,
  minimum height=6mm,
  very thick,
]

\tikzstyle{functnode} = [
  rectangle,
  minimum width=1cm,
  minimum height=5mm,
]

\tikzstyle{idnode} = [
  sortnode,
  draw=deepblue,
  fill=cyan!50,
  text=deepblue,  
]

\tikzstyle{artnode} = [
  sortnode,
  draw=deepblue,
  fill=yellow!50,
  text=deepblue,  
]

\tikzstyle{valnode} = [
  sortnode,
  densely dashed,
  draw=violet,
  fill=magenta!50,
  text=violet,
]

\tikzstyle{f} = [
  thick,
]

\tikzstyle{fd} = [
  f,
  -angle 90,
]

\tikzstyle{relation}=[rectangle split, rectangle split parts=#1, rectangle split part align=base, draw, anchor=center, align=center, text height=3mm, font=\bfseries, text centered]

\begin{document}

\title{Combined Covers and Beth Definability (Extended Version)}


\maketitle


\begin{abstract}
Uniform interpolants were largely studied in non-classical propositional logics since the nineties, and their connection to model completeness was pointed out in the literature. A successive parallel research line inside the automated reasoning community investigated uniform quantifier-free interpolants (sometimes referred to as ``covers'') in first-order theories. In this paper, we investigate cover transfer to theory combinations in the disjoint signatures case. We prove that, for convex theories, cover algorithms can be transferred to theory combinations under the same hypothesis needed to transfer quantifier-free interpolation (i.e., the
equality interpolating property, aka strong amalgamation property). The key feature of our algorithm relies on the extensive usage of the Beth definability property for primitive fragments to convert implicitly defined variables into their explicitly defining terms. In the non-convex case, we show by a counterexample that covers may not exist in the combined theories, 
even in case combined quantifier-free interpolants do exist. However, we exhibit a cover transfer algorithm operating also in the non-convex case for special kinds of theory combinations; these combinations (called `tame combinations') concern multi-sorted theories arising in many model-checking applications (in particular, the ones oriented to verification of data-aware processes).
\end{abstract}

\section{Introduction}

Uniform interpolants were originally studied in the context of non-classical logics, starting from the pioneering work by Pitts~\cite{pitts}. We briefly recall what uniform interpolants are; we fix a logic or a theory $T$ and a suitable fragment (propositional, first-order quantifier-free, etc.) of its language $L$.
Given an $L$-formula $\phi(\ux, \uy)$ (here $\ux,\uy$ are the  variables occurring free in $\phi$), a \emph{uniform interpolant} of $\phi$ (w.r.t.~$\uy$) is a formula $\phi'(\ux)$ where only the $\ux$ occur free, and that satisfies the following two properties:
\begin{inparaenum}[\it (i)]
	\item $\phi(\ux, \uy)\vdash_T \phi'(\ux)$; 
	\item for any further $L$-formula $\psi(\ux, \uz)$ such that $\phi(\ux, \uy) \vdash_T \psi(\ux, \uz)$, we have $\phi'(\ux) \vdash_T \psi(\ux, \uz)$. 
 \end{inparaenum}
Whenever uniform interpolants exist, one can compute an interpolant for an entailment like $\phi(\ux, \uy) \vdash_T \psi(\ux, \uz)$ in a way that is \emph{independent} of $\psi$.

The existence of uniform interpolants is an exceptional phenomenon, which is however not so infrequent; it has been extensively studied in non-classical logics starting from the nineties, as witnessed by a large literature (a non-exhaustive list includes
\cite{shavrukov,visser,GZ,GZsl,GZjsl,nf,bilkova,M,MK}).  
The main results from the above papers are that uniform interpolants exist for intuitionistic logic and for some modal systems (like  the G\"odel-L\"ob system and the S4.Grz system); they do not exist for instance in $S4$ and $K4$, whereas for the basic modal system $K$ they exist for the local consequence relation but not for the global consequence relation. 
The connection between uniform interpolants and model completions (for equational theories axiomatizing the varieties corresponding to propositional logics) was first stated in~\cite{GZapal} and further developed in~\cite{GZ,M,MK}.

In the last decade, also the automated reasoning community developed an increasing interest in uniform interpolants, with particular focus on quantifier-free fragments of first-order theories. This is witnessed by various talks and drafts by D. Kapur presented in 
many
conferences and workshops (FloC 2010, ISCAS 2013-14, SCS 2017~\cite{kapur}), as well as by the paper~\cite{GM} by Gulwani and Musuvathi in ESOP 2008. In this last paper uniform interpolants were renamed as \emph{covers}, a terminology we shall adopt in this paper too.  In these contributions, examples of cover computations were supplied and also some algorithms were sketched. The first formal \emph{proof} about existence of covers in \EUF\ was however published by the present authors only in~\cite{cade19}; such a proof was equipped with powerful semantic tools 
(the Cover-by-Extensions Lemma~\ref{lem:cover} below) coming from the connection to model-completeness, as well as with an algorithm relying on a constrained variant of the Superposition Calculus (two 
simpler 
algorithms are studied in~\cite{GGK}).
 The usefulness of covers in model checking was already stressed in~\cite{GM} and further motivated by our recent line of research on the verification of data-aware processes \cite{CGGMR19,BPM19,MSCS20,ARCADE}. Notably, it is also operationally mirrored in the \textsc{MCMT} \cite{mcmt} implementation since version 2.8. Covers (via quantifier elimination in model completions and hierarchical reasoning)  play an important role in symbol elimination problems in theory extensions, as witnesssed in the comprehensive paper~\cite{viorica18} and in related papers~\cite{viorica19} studying invariant synthesis in model checking applications.

An important question suggested by the applications
is the cover transfer problem for combined theories: for instance, when modeling and verifying data-aware processes, it is natural to consider the combination of different theories, such as the theories accounting for the read-write and read-only data storage of the process as well as those for the elements stored therein~\cite{CGGMR19,cade19,ARCADE}.
Formally, the cover transfer problem can be stated as follows:
\emph{by supposing that covers exist in theories $T_1, T_2$, under which conditions do they exist also in the combined theory  $T_1\cup T_2$?}
In this paper we show that the answer is affirmative in the disjoint signatures convex case, using the same hypothesis (that is, the equality interpolating condition) under which quantifier-free interpolation transfers. Thus, for convex theories we essentially obtain a necessary and sufficient condition, in the precise sense captured by Theorem~\ref{thm:necessary} below. 
We also prove that if convexity fails, the non-convex equality interpolating property~\cite{bgr-acmtocl} may not be sufficient to ensure the cover transfer property. As a witness for this, we show that \EUF\ combined with integer difference logic or with linear integer arithmetics constitutes a counterexample.

The main tool employed in our combination result is the \emph{Beth definability theorem for primitive formulae} (this theorem has been shown to be equivalent to the equality interpolating condition in~\cite{bgr-acmtocl}). In  order to design a combined cover algorithm, we exploit the equivalence between implicit and explicit definability that is supplied by the Beth theorem. Implicit definability is reformulated, via covers for input theories, at the quantifier-free level. Thus, the combined cover algorithm guesses the implicitly definable variables, then eliminates them via explicit definability, and finally uses the component-wise input cover algorithms to eliminate the remaining (non implicitly definable) variables.
The identification and the elimination of the implicitly defined variables via explicitly defining terms is an essential step towards the correctness of the combined cover algorithm: when computing a cover of a formula $\phi(\ux, \uy)$ (w.r.t.~$\uy$), the variables $\ux$ are (non-eliminable) parameters, and those variables among the $\uy$ that are implicitly definable \emph{need to be discovered and treated in the same way as the parameters $\ux$}. Only after this preliminary step (Lemma~\ref{lem:toterminal} below), the input cover algorithms can be suitably exploited (Proposition~\ref{prop:terminal} below). 

The combination result we obtain is quite strong, as it is a typical `black box' combination result: it applies not only to theories used in verification (like the combination of real arithmetics with $\EUF$), but also in other contexts. For instance, since the theory $\mathcal B$ of Boolean algebras satisfies our hypotheses (being model completable and strongly amalgamable~\cite{GG18}), we get that uniform interpolants exist in the combination of $\mathcal B$ with $\EUF$. The latter is the equational theory algebraizing 
the
basic non-normal classical modal logic
system $\bf E$ from~\cite{segerberg} 
(extended to $n$-ary modalities). 
Notice that this result must be contrasted with the case of many systems of Boolean algebras with operators where existence of uniform interpolation fails~\cite{MK} (recall that operators on a Boolean algebra are not just arbitrary functions, but are required to be monotonic and also to preserve either joins or meets in each coordinate).

As a last important comment on related work, it is worth mentioning that Gulwani and Musuvathi in~\cite{GM} also have a combined cover algorithm for convex, signature disjoint theories. Their algorithm looks quite different from ours; apart from the fact that a full correctness and completeness proof for such an algorithm has never been published, we underline that our algorithm is rooted on different hypotheses. In fact, we only need the equality interpolating condition and we show that this hypothesis is not only sufficient, but also necessary for cover transfer in convex theories; consequently, our result is 
formally stronger. The equality interpolating condition was known to the authors of~\cite{GM} 
(but not even mentioned in their paper~\cite{GM}): in fact, it was introduced by one of them some years before \cite{ym}. 
The  equality interpolating condition was then extended to the non convex case in~\cite{bgr-acmtocl}, where it was also semantically characterized via the strong amalgamation property.


The paper is organized as follows: after some preliminaries in Section~\ref{sec:prelim}, the crucial Covers-by-Extensions Lemma and the relationship between covers and model completions from~\cite{cade19} are recalled in Section~\ref{sec:covers}. In Section~\ref{sec:strongamalgamation}, we present some preliminary results on interpolation and Beth definability that are instrumental to our machinery. After some useful facts about convex theories in Section~\ref{sec:convex}, we introduce the combined cover algorithms for the convex case and we prove its correctness in Section~\ref{sec:combined}; we also present a detailed example of application of the combined algorithm in case of the combination of \EUF\ with linear real arithmetic, and we show that the equality interpolating condition is necessary (in some sense) for combining covers. In Section~\ref{sec:nonconvex} we exhibit a counteraxample to the existence of combined covers in the non-convex case. Finally, in Section~\ref{sec:value} we prove that for the `tame' multi-sorted theory combinations used in our database-driven applications, covers existence transfers to the combined theory under only the stable infiniteness requirement for the shared sorts. Section~\ref{sec:conclusions} is devoted to the conclusions and discussion of future work. The current paper is the extended version of \cite{IJCAR20}.

\section{Preliminaries}
\label{sec:prelim}

We adopt the usual first-order syntactic notions of signature, term,
atom, (ground) formula,  and so on; our signatures are always \emph{finite} or \emph{countable}
and include equality. To avoid considering limit cases, we assume that signatures always contain at least an individual constant.
%
We compactly represent a tuple $\tup{x_1,\ldots,x_n}$ of variables as $\ux$. The notation $t(\ux), \phi(\ux)$ means that the term $t$, the formula $\phi$ has free variables included in the tuple $\ux$. This tuple is assumed to be formed by \emph{distinct variables}, thus we underline that when we write e.g. $\phi(\ux, \uy)$, we mean that the tuples $\ux, \uy$ are made of distinct variables that are also disjoint from each other.
%

A formula is said to be \emph{universal} (resp., \emph{existential}) if it has the form $\forall \ux (\phi(\ux))$ (resp., $\exists \ux (\phi(\ux))$), where $\phi$ is  quantifier-free. Formulae with no free variables are called \emph{sentences}. 
On the semantic side, we use the standard notion of $\Sigma$-structure $\cM$ and of truth of a formula in a $\Sigma$-structure under a free variables assignment. 
The support of $\cM$ is denoted as 
$\vert \cM\vert$.
The interpretation of a 
(function, predicate) 
symbol $\sigma$ in $\cM$ is denoted $\sigma^{\cM}$.

A \emph{$\Sigma$-theory} $T$ is a set of $\Sigma$-sentences; a \emph{model}  of $T$ is a $\Sigma$-structure $\cM$ where all sentences in $T$ are true.
	 We use the standard notation $T\models \phi$ to say that $\phi$ is true in all models of $T$ for every assignment to the variables occurring free in $\phi$. We say that $\phi$ is \emph{$T$-satisfiable} iff there is a model $\cM$ of $T$ and an assignment to the variables occurring free in $\phi$ making $\phi$ true in $\cM$.
	 
	 We now focus on the constraint satisfiability problem and quantifier elimination for a theory $T$.
	 A $\Sigma$-formula $\phi$ is a $\Sigma$-\emph{constraint} (or just a constraint) iff it is a conjunction of literals.
	The \emph{constraint satisfiability problem} for $T$ is the following: we are given a constraint 
	$\phi(\ux)$
	and we are asked whether there exist  a model $\cM$ of $T$ and an assignment $\cI$ to the free variables $\ux$ such that $\cM, \cI \models \phi(\ux)$. 	
	 A theory $T$ has \emph{quantifier elimination} iff for every formula $\phi(\ux)$ in the signature of $T$ there is a quantifier-free formula 
	$\phi'(\ux)$ such that $T\models \phi(\ux)\leftrightarrow \phi'(\ux)$.
	Since we are in a
	computational logic context, when we speak of quantifier elimination, we assume that it is effective, namely that it comes
	with an algorithm for computing 
	$\phi'$ out of $\phi$.
	It is well-known that quantifier elimination holds in case we can eliminate quantifiers from \emph{primitive} formulae, i.e., formulae of the kind $\exists \uy \,\phi(\ux, \uy)$, with $\phi$ a constraint.
	%

 We recall also some further basic notions.
 %
 %
 Let $\Sigma$ be a first-order signature. The signature
 obtained from $\Sigma$ by adding to it a set $\ua$ of new constants
 (i.e., $0$-ary function symbols) is denoted by $\Sigma^{\ua}$.
Analogously, given a $\Sigma$-structure $\cM$, the signature $\Sigma$ can be expanded to a new signature $\Sigma^{|\cM|}:=\Sigma\cup \{\bar{a}\ |\ a\in |\cM| \}$ by adding a set of new constants $\bar{a}$ (the \textit{name} for $a$), one for each element $a$ in the support of $\cM$, with the convention that two distinct elements are denoted by different ``name'' constants. $\cM$ can be expanded to a $\Sigma^{|\cM|}$-structure $\overline{\cM}:=(\cM, a)_{a\in \vert\cM\vert}$ just interpreting the additional constants over the corresponding elements. From now on, when the meaning is clear from the context, we will freely use the notation  $\cM$ and  $\overline{\cM}$ interchangeably: in particular, given a $\Sigma$-structure
$\cM$ 
and a $\Sigma$-formula $\phi(\ux)$ with free variables that are all in $\ux$, we will write, by abuse of notation, 
$\cM\models \phi(\ua)$ instead of $\overline{\cM}\models \phi(\bar{\ua})$.

A {\it $\Sigma$-homomorphism} (or, simply, a
homomorphism) between two $\Sigma$-structu\-res $\cM$ and
$\cN$ is a map $\mu: \vert \cM \vert \lra \vert \cN\vert $ among the
support sets $\vert \cM \vert $ of $\cM$ and $\vert \cN \vert$  of $\cN$ satisfying the condition
$(\cM \models \varphi \quad \Rightarrow \quad \cN \models \varphi)$
for all $\Sigma^{\vert \cM\vert}$-atoms $\varphi$ ($\cM$ is regarded as a
$\Sigma^{\vert \cM\vert}$-structure, by interpreting each additional constant $a\in
\vert \cM\vert $ into itself and $\cN$ is regarded as a $\Sigma^{\vert \cM\vert}$-structure by
interpreting each additional constant $a\in \vert \cM\vert $ into $\mu(a)$). 
In case the last condition 
holds for all $\Sigma^{|\cM|}$-literals,
the homomorphism $\mu$ is said to be an {\it
	embedding} and if it holds for all 
first order
formulae, the embedding $\mu$ is said to be {\it
	elementary}. 
If  $\mu: \cM \lra \cN$ is an embedding which is just the
identity inclusion $\vert \cM\vert\subseteq\vert \cN\vert$, we say that $\cM$ is a {\it
	substructure} of $\cN$ or that $\cN$ is an {\it extension} of
$\cM$. Universal theories can be characterized as those theories $T$ having the property that if $\cM\models T$ and $\cN$ is a substructure of $\cM$, then $\cN\models T$ (see~\cite{CK}). If $\cM$ is a structure and $X\subseteq \vert \cM\vert$, then there is the smallest substructure of $\cM$ including $X$ in its support; this is called the substructure \emph{generated by $X$}. If $X$ is the set of elements of a finite tuple $\ua$, then the substructure generated by $X$ has in its support precisely the $b\in \vert\cM\vert$ such that 
$\cM\models b=t(\ua)$ for some term $t$.

Let $\cM$ be a $\Sigma$-structure. The \textit{diagram} of $\cM$, written $\Delta_{\Sigma}(\cM)$ (or just $\Delta(\cM)$), is the set of ground $\Sigma^{|\cM|}$-literals 
that are true in $\cM$. 
%
An easy but important result, called 
\emph{Robinson Diagram Lemma}~\cite{CK},
says that, given any $\Sigma$-structure $\cN$,  the embeddings $\mu: \cM \longrightarrow \cN$ are in bijective correspondence with
expansions of $\cN$ 
to   $\Sigma^{\vert \cM\vert}$-structures which are models of 
$\Delta_{\Sigma}(\cM)$. The expansions and the embeddings are related in the obvious way: $\bar a$ is interpreted as $\mu(a)$.

\section{Covers and Model Completions}\label{sec:covers}

We report the notion of \emph{cover} taken from~\cite{GM} and also the basic results proved in~\cite{cade19}. 
Fix a  theory $T$ and an existential formula $\exists \ue\, \phi(\ue, \uy)$; call a \emph{residue} of $\exists \ue\, \phi(\ue, \uy)$ any quantifier-free formula belonging to the set of quantifier-free formulae  $Res(\exists \ue\, \phi)=\{\theta(\uy, \uz)\mid T \models \phi(\ue, \uy) \to \theta(\uy, \uz)\}$. 
A quantifier-free formula $\psi(\uy)$ is said to be a \emph{$T$-cover} (or, simply, a \emph{cover}) of $\exists \ue\, \phi(\ue,\uy)$ iff  $\psi(\uy)\in Res(\exists \ue\, \phi)$ and $\psi(\uy)$ implies (modulo $T$) all the other formulae in $Res(\exists \ue\, \phi)$.
The following ``cover-by-extensions'' Lemma~\cite{cade19} (to be widely used throughout the paper) supplies a semantic counterpart to the notion of a cover: 

\begin{lemma}[Cover-by-Extensions]\label{lem:cover} A formula $\psi(\uy)$ is a $T$-cover of $\exists \ue\, \phi(\ue, \uy)$ iff 
it satisfies the following two conditions:
\begin{inparaenum}[(i)]
\item $T\models  \forall \uy\,( \exists \ue\,\phi(\ue, \uy) \to \psi(\uy))$;
\item for every model $\cM$ of $T$, for every tuple of  elements $\ua$ from the support of $\cM$ such that $\cM\models \psi(\ua)$ it is possible to find
  another model $\cN$ of $T$ such that $\cM$ embeds into $\cN$ and $\cN\models \exists \ue \,\phi(\ue, \ua)$.
\end{inparaenum}
\end{lemma}

We underline that, since our language is at most countable, we can assume that the models $\cM$, $\cN$ from (ii) above are at most countable too, by a L\"owenheim-Skolem argument.

We say that a theory $T$ has \emph{uniform quantifier-free interpolation} iff every existential formula $\exists \ue\, \phi(\ue,\uy)$ (equivalently, every primitive formula $\exists \ue\, \phi(\ue,\uy)$) has a $T$-cover.

It is clear that if $T$ has uniform quantifier-free interpolation, then it has ordinary \emph{quantifier-free interpolation}~\cite{bgr-acmtocl}, in the sense that if we have $T\models \phi(\ue, \uy)\to \phi'(\uy, \uz)$ (for quantifier-free formulae $\phi, \phi'$), then there is a quantifier-free formula $\theta(\uy)$ such that  $T\models \phi(\ue, \uy)\to \theta(\uy)$ and $T\models \theta(\uy)\to \phi'(\uy, \uz)$. In fact, if $T$ has uniform quantifier-free interpolation, then the interpolant $\theta$ is independent on $\phi'$ (the  same $\theta(\uy)$ can be used as interpolant for all entailments $T\models \phi(\ue, \uy)\to \phi'(\uy, \uz)$, varying $\phi'$).

 We say that a \emph{universal} theory $T$ has a \emph{model completion} iff there is a stronger theory $T^*\supseteq T$ (still within the same signature $\Sigma$ of $T$) such that (i) every $\Sigma$-constraint that is satisfiable
	in a model of $T$ is satisfiable in a model of $T^*$; (ii) $T^*$ eliminates quantifiers.
	Other equivalent definitions are possible~\cite{CK}: for instance, (i) is equivalent to the fact that $T$ and $T^*$ prove the same universal formulae or again to the fact that every model of $T$ can be embedded into a model of $T^*$.
	We recall that the model completion, if it exists, is unique and that its existence implies the quantifier-free interpolation property for $T$~\cite{CK}
	(the latter can be seen directly or via the correspondence between quantifier-free interpolation and amalgamability, see~\cite{bgr-acmtocl}).
	%
	
	A close relationship between model completion and uniform interpolation emerged in the area of propositional logic (see the book~\cite{GZ}) and can be formulated roughly as follows. It is well-known that most propositional calculi, via Lindembaum constructions, can be algebraized: the algebraic analogue of classical logic are Boolean algebras, the algebraic analogue of intuitionistic logic are Heyting algebras, the algebraic analogue of modal calculi are suitable variaties of modal agebras, etc. Under suitable hypotheses, it turns out that a propositional logic has uniform interpolation (for the global consequence relation) iff the equational theory axiomatizing the corresponding variety of algebras has a model completion~\cite{GZ}.
	In  the context of first order theories, we 
	prove an even more direct connection:

\begin{theorem}\label{prop:qe} Suppose that $T$ is a universal theory. Then $T$ has a model completion $T^*$ iff $T$ has uniform quantifier-free interpolation. If this happens, 
  $T^*$ is axiomatized by the infinitely many sentences
  $\forall \uy \,(\psi(\uy) \to \exists \ue\, \phi(\ue, \uy))$,
  where $\exists \ue\, \phi(\ue, \uy)$ is a  primitive formula and $\psi$ is a cover of it.
\end{theorem}
%
The proof (via Lemma~\ref{lem:cover}, by iterating  a chain construction) is in~\cite{MSCS20} (see also~\cite{new}). 

\section{Equality Interpolating Condition and Beth Definability}\label{sec:strongamalgamation} %

We report here some definitions and results 
we need 
concerning combined quantifier-free interpolation. 
Most definitions and result come from~\cite{bgr-acmtocl}, but are simplified here 
 because we restrict them to the case of universal convex theories. Further information on the semantic side is supplied in Appendix~\ref{sec:approofs}.

A theory $T$ is \emph{stably infinite} iff every $T$-satisfiable
constraint is satisfiable in
an infinite model of $T$.
The following Lemma comes from a compactness argument 
(see Appendix~\ref{sec:approofs} for a proof):
\begin{lemma}\label{lem:si}
 If $T$ is stably infinite, then every finite or countable model $\cM$ of $T$ can be embedded in a model $\cN$ of $T$ such
 that  $\vert \cN\vert \setminus \vert \cM\vert$ is countable.
\end{lemma}

A theory $T$ is \emph{convex} iff 
for every constraint $\delta$, if $T\vdash\delta \to
\bigvee_{i=1}^n x_i=y_i$ then $T\vdash\delta\to x_i=y_i$ holds for some
$i\in \{1, ..., n\}$.
A convex theory $T$ is `almost' stably infinite in the sense that it can be shown that every constraint which is $T$-satisfiable in a $T$-model whose support has at least two elements is satisfiable also in an infinite $T$-model. The one-element model can be used to build counterexamples, though: e.g.,
the theory of Boolean algebras is convex (like any other universal Horn theory) but the constraint $x=0 \wedge x=1$ is only satisfiable in the degenerate one-element Boolean algebra. Since we take into account these limit cases, we do not  assume that convexity implies stable infiniteness. 

\begin{definition}
  \label{def:YM}
   A convex universal theory $T$ is \emph{equality interpolating}  iff 
     for every pair $y_1, y_2$ of variables and for every pair
       of \emph{constraints} $\delta_1(\ux,\uz_1, y_1),
       \delta_2(\ux,\uz_2,y_2)$ such that
       \begin{equation}
         \label{eq:ym_ant}
         T\vdash\delta_1(\ux, \uz_1,y_1)\wedge
         \delta_2(\ux,\uz_2, y_2)\to
         y_1= y_2
       \end{equation}
       there exists a term $t(\ux)$ such that
       \begin{equation}
         \label{eq:ym_cons}
         T\vdash
         \delta_1(\ux, \uz_1, y_1)\wedge
         \delta_2(\ux, \uz_2, y_2)\to
         y_1= t(\underline{x})  \wedge  y_2= t(\underline{x}).
       \end{equation}
\end{definition}

\begin{theorem}\cite{ym,bgr-acmtocl}
 \label{thm:sub-amalgamation_transfer}
 Let $T_1$ and $T_2$ be two universal, convex, stably infinite theories over disjoint
 signatures $\Sigma_1$ and $\Sigma_2$.  If both $T_1$ and $T_2$ are equality interpolating and have quantifier-free interpolation property,  then so does $T_1\cup T_2$. 
\end{theorem}

There is a converse of the previous result; for a signature $\Sigma$, let us call $\EUF(\Sigma)$ the pure equality theory over the signature $\Sigma$ (this theory is equality interpolating and has the quantifier-free interpolation property).

\begin{theorem}\cite{bgr-acmtocl}
  \label{prop:strong_amalgamation_needed} 
  Let $T$ be a stably infinite, universal, convex theory 
  admitting quantifier-free interpolation 
  and let
  $\Sigma$ be a signature disjoint from the signature of $T$
  containing at least a unary predicate symbol.  Then, $T\cup
  \EUF(\Sigma)$ has quantifier-free interpolation iff $T$ is equality interpolating.
\end{theorem}

In~\cite{bgr-acmtocl} the above definitions and results are extended to the non-convex case and a long list of universal  quantifier-free interpolating and equality interpolating theories is given. The list includes $\EUF(\Sigma)$, recursive data theories, as well as linear  arithmetics. For linear arithmetics (and fragments of its), it is essential to make a very careful choice of the signature, see again~\cite{bgr-acmtocl} (especially Subsection 4.1) for details.  All the above theories admit a model completion (which coincides with the theory itself in case the theory admits quantifier elimination).

The equality interpolating property in a  theory $T$ can be equivalently characterized using Beth definability as follows. Consider a primitive formula $\exists \uz \phi(\ux, \uz,y)$ (here $\phi$ is a conjunction of literals); we say that
$\exists \uz\, \phi(\ux, \uz,y)$ \emph{implicitly defines $y$} in $T$ iff the formula
\begin{equation}\label{eq:bethimpl}
 \forall y \,\forall y'\;( \exists \uz \phi(\ux, \uz,y)\wedge \exists \uz \phi(\ux, \uz,y') \to y= y')
\end{equation}
is $T$-valid.
We say that $\exists \uz \phi(\ux, \uz,y)$ \emph{explicitly defines $y$} in $T$ iff there is a term $t(\ux)$ such that the formula
\begin{equation}\label{eq:bethexpl}
 \forall y \;( \exists \uz \phi(\ux, \uz,y) \to y= t(\ux))
\end{equation}
 is $T$-valid.

 For future use, we notice that,
  by trivial logical manipulations, the formulae~\eqref{eq:bethimpl} and~\eqref{eq:bethexpl} are logically equivalent to
\begin{equation}\label{eq:bethimpl1}
  \forall y\forall \uz \forall y'\forall \uz'( \phi(\ux, \uz,y)\wedge \phi(\ux, \uz',y') \to y= y')~~~.
\end{equation}
and to
\begin{equation}\label{eq:bethexpl1}
 \forall y \forall \uz (\phi(\ux, \uz,y) \to y= t(\ux))
\end{equation}
respectively (we shall use such equivalences without explicit mention).

We say that a theory $T$ has the \emph{Beth definability property for primitive formulae} iff whenever a primitive formula
$\exists \uz\, \phi(\ux, \uz,y)$ implicitly defines  the variable $y$ then it also explicitly defines it.

\begin{theorem}\cite{bgr-acmtocl}
  \label{thm:beth}
  A convex theory $T$ having quan\-ti\-fier-free interpolation
   is equality interpolating iff it has the Beth definability property for primitive formulae.
\end{theorem}

\begin{proof}
 We recall the easy proof of the left-to-right side (this is the only side we need in this paper). Suppose that $T$ is equality interpolating and that  
 \begin{equation*}
 T\vdash \phi(\ux, \uz,y)\wedge \phi(\ux, \uz',y') \to y= y'~~;
\end{equation*}
then there is a term $t(\ux)$ such that
\begin{equation*}
 T\vdash \phi(\ux, \uz,y)\wedge \phi(\ux, \uz',y') \to y=t(\ux) \wedge  y'=t(\ux)~~.
\end{equation*}
Replacing $\uz',y'$ by $\uz,y$ via a substitution, we get precisely~\eqref{eq:bethexpl1}. 
\end{proof}

\section{Convex Theories}\label{sec:convex}

We now collect some useful facts concerning convex theories. 
We fix for this section a \emph{convex, stably infinite, equality interpolating universal theory $T$ admitting a model completion $T^*$}.
We let $\Sigma$ be the signature of $T$. We fix also 
\emph{a  $\Sigma$-constraint $\phi(\ux, \uy)$}, 
where we assume that $\uy=y_1, \dots, y_n$ (recall that the tuple $\ux$ is disjoint from the tuple $\uy$ according to our conventions from Section~\ref{sec:prelim}). 

For $i=1, \dots, n$, we let the formula $\mathtt{ImplDef}_{\phi,y_i}^T(\ux)$ be the quantifier-free formula equivalent in $T^*$ to the formula
\begin{equation}\label{eq:def}
 \forall \uy\, \forall \uy' (\phi(\ux, \uy) \wedge \phi(\ux, \uy')\to y_i= y'_i)
\end{equation}
where the $\uy'$ are renamed copies of the $\uy$. Notice that the variables occurring free in $\phi$ are $\ux, \uy$, whereas only the $\ux$ occur free in $\mathtt{ImplDef}_{\phi,y_i}^T(\ux)$ (the variable $y_i$ is among the $\uy$ and does not occur free in $\mathtt{ImplDef}_{\phi,y_i}^T(\ux)$): these facts coming from our notational conventions are crucial and should be  
kept in mind when reading this and next section. 
%
%
The following semantic technical lemma is proved in Appendix~\ref{sec:approofs}:

\begin{lemma}\label{lem1}
 Suppose that we are given a  model $\cM$ of $T$ and elements $\ua$ from the support of $\cM$ such that  $\cM\not\models \mathtt{ImplDef}_{\phi,y_i}^T(\ua)$ for all $i=1, \dots,n$. Then there exists an extension $\cN$ of $\cM$ such that 
 for some $\ub\in \vert\cN\vert \setminus \vert \cM\vert$ we have $\cN\models \phi(\ua, \ub)$. 
\end{lemma}

The following Lemma supplies terms which will be used as ingredients in our combined covers algorithm:

\begin{lemma}\label{lem2}
 Let $L_{i1}(\ux)\vee \cdots \vee L_{ik_i}(\ux)$ be the disjunctive normal form (DNF) of $\mathtt{ImplDef}_{\phi,y_i}^T(\ux)$. Then, for every $j=1, \dots, k_i$, there is a $\Sigma(\ux)$-term $t_{ij}(\ux)$
 such that
 \begin{equation}\label{eq:term}
  T\vdash L_{ij}(\ux)\wedge \phi(\ux,\uy)\to y_i=t_{ij}~~.
 \end{equation}
 As a consequence, a  formula of the kind $\mathtt{ImplDef}_{\phi,y_i}^T(\ux) \wedge \exists\uy\, (\phi(\ux, \uy)\wedge \psi)$
 is equivalent (modulo $T$) to the formula
 \begin{equation}\label{eq:terms}
  \bigvee_{j=1}^{k_i} \exists \uy\; (y_i=t_{ij} \wedge L_{ij}(\ux)   \wedge \phi(\ux, \uy)\wedge \psi)~~.
 \end{equation}

\end{lemma}
\begin{proof} We have that $(\bigvee_j L_{ij})\leftrightarrow \mathtt{ImplDef}_{\phi,y_i}^T(\ux)$ is a tautology, hence
 from the definition of $\mathtt{ImplDef}_{\phi,y_i}^T(\ux)$, we have that 
 $$
 T^*\vdash L_{ij}(\ux)\to \forall \uy\, 
 \forall \uy' (\phi(\ux, \uy) \wedge \phi(\ux, \uy')\to y_i= y'_i)~~;
 $$
 however this formula is trivially equivalent to a universal formula ($L_{ij}$ does not depend on $\uy,\uy'$), 
  hence since $T$ and $T^*$ prove the same universal formulae, we get 
 $$ 
 T\vdash L_{ij}(\ux) \wedge 
\phi(\ux, \uy) \wedge \phi(\ux, \uy')\to y_i= y'_i~~.
 $$
 Using Beth definability property (Theorem~\ref{thm:beth}),
 we get~\eqref{eq:term}, as required, for some terms $t_{ij}(\ux)$.
Finally,
the second claim of the lemma
follows from~\eqref{eq:term} by trivial logical manipulations. 
\end{proof}

In all our concrete examples, the theory $T$ has decidable quantifier-free fragment (namely it is decidable whether a quantifier-free formula is a logical consequence of $T$ or not), thus the terms $t_{ij}$ mentioned in Lemma~\ref{lem2} can be computed just by enumerating all possible $\Sigma(\ux)$-terms: the computation terminates, because the above proof shows that the appropriate terms always exist. However, this is terribly inefficient and, from a practical point of view, one needs to have at disposal dedicated algorithms to find the required equality interpolating terms.
For some common theories (\EUF, Lisp-structures, linear real arithmetic), such algorithms are designed in~\cite{ym}; in~\cite{bgr-acmtocl} [Lemma 4.3 and Theorem 4.4], the algorithms for computing equality interpolating terms are connected to quantifier elimination algorithms in the case of universal theories admitting quantifier elimination. Still,  an extensive investigation on te topic seems to be  missed in the SMT literature.

\section{The Convex Combined Cover Algorithm}\label{sec:combined}

Let us now fix two theories $T_1, T_2$ over disjoint signatures $\Sigma_1, \Sigma_2$.
We assume that both of them satisfy the assumptions from the previous section, meaning that they are convex, stably infinite, equality interpolating, universal and admit  model completions $T^*_1, T^*_2$ respectively. We shall supply a cover algorithm for $T_1\cup T_2$ (thus proving that $T_1\cup T_2$ has a model completion too).

We need to compute a cover for $\exists \ue\,\phi(\ux, \ue)$, where $\phi$ is a conjunction of $\Sigma_1\cup\Sigma_2$-literals. By applying rewriting purification steps like
$$
\phi \Longrightarrow
\exists d\, (d=t \wedge \phi(d/t))
$$
(where $d$ is a fresh variable and $t$ is a pure term, i.e. it is  either a $\Sigma_1$- or a $\Sigma_2$-term), 
we can assume that our formula $\phi$ is of the kind $\phi_1\wedge \phi_2$, where $\phi_1$ is a $\Sigma_1$-formula and $\phi_2$ is a $\Sigma_2$-formula. 
Thus we need to compute a cover for a formula of the kind
\begin{equation}\label{eq:initial}
 \exists \ue\,(\phi_1(\ux, \ue)\wedge \phi_2(\ux, \ue)),
\end{equation}
where $\phi_i$ is a conjunction of $\Sigma_i$-literals ($i=1,2$). We also assume that both $\phi_1$ and $\phi_2$ contain the literals $e_i\neq e_j$ (for $i\neq j$) as a conjunct: this 
can be achieved by guessing a partition of the $\ue$ and by replacing each $e_i$ with the representative element of its equivalence class.

\begin{remark}
 It is not clear whether this preliminary guessing step can be avoided. In fact, Nelson-Oppen~\cite{NO} combined satisfiability  for \emph{convex} theories does not need it; however, combining covers algorithms is a more complicated problem than combining mere satisfiability algorithms and for technical reasons related to the correctness and completeness proofs below, we were forced to introduce guessing at this step.
\end{remark}

 To manipulate formulae, our algorithm employs acyclic explicit definitions as follows.
When we write $\mathtt{ExplDef}(\uz, \ux)$ (where $\uz, \ux$ are tuples of distinct variables), we mean any  formula of the kind (let $\uz:=z_1 \dots, z_m$)
$$
\bigwedge_{i=1}^m z_i =t_i(z_1,\dots,z_{i-1},\ux)
$$
where the term $t_i$ is pure 
(i.e. it is a $\Sigma_i$-term)
and only the variables $z_1, \dots, z_{i-1}, \ux$ can occur in it. When we assert a formula like $\exists \uz\;(\mathtt{ExplDef}(\uz, \ux)\wedge \psi(\uz, \ux))$, we are in fact in  the condition of recursively eliminating the variables $\uz$ from it via terms containing only the parameters $\ux$ 
(the 'explicit definitions' $z_i =t_i$ are in fact arranged acyclically).

A \emph{working formula} is a formula of the kind 
\begin{equation}\label{eq:working}
 \exists \uz \,(\mathtt{ExplDef}(\uz, \ux) \wedge \exists \ue \, (\psi_1(\ux, \uz, \ue) \wedge \psi_2(\ux,\uz, \ue)))~~,
\end{equation}
where $\psi_1$ is a conjunction of $\Sigma_1$-literals and $\psi_2$ is a conjunction of $\Sigma_2$-literals.
The variables $\ux$ are called \emph{parameters}, the variables $\uz$ are called \emph{defined variables} and the variables $\ue$ \emph{(truly) existential variables}. The parameters do not change during the execution of the algorithm. 
We assume that $\psi_1, \psi_2$ in a working formula~\eqref{eq:working} always contain the literals $e_i\neq e_j$ (for distinct $e_i, e_j$ from $\ue$) as a conjunct.

In our starting formula~\eqref{eq:initial}, there are no defined variables. However, if via some syntactic check it happens that some of the existential variables can be 
 recognized as defined, then it is useful to display them as such (this observation may avoid redundant cases - leading to inconsistent disjuncts - in the computations below).

A working formula like~\eqref{eq:working} is said to be \emph{terminal} iff for every existential variable $e_i\in \ue$ we have that
\begin{equation}\label{eq:ter}
 T_1\vdash \psi_1\to \neg\mathtt{ImplDef}_{\psi_1,e_i}^{T_1}(\ux,\uz) ~~{\rm and}~~T_2\vdash \psi_2\to \neg\mathtt{ImplDef}_{\psi_2,e_i}^{T_2}(\ux,\uz)~~~.
\end{equation}
Roughly speaking, we can say that in a terminal working formula, all variables which are not parameters are either explicitly definable or recognized as not implicitly definable by both theories; of course, a working formula with no existential variables is terminal.

\begin{lemma}\label{lem:toterminal}
 Every working formula is equivalent (modulo $T_1\cup T_2$) to a disjunction of terminal working formulae.
\end{lemma}

\begin{proof} We only sketch the proof of this Lemma  (see the Appendix~\ref{sec:approofs} for full details), by describing the algorithm underlying it. 
 To compute the required terminal working formulae, it is sufficient to apply the following non-deterministic procedure (the output is  the disjunction of all possible outcomes). 
The non-deterministic procedure applies one of the following alternatives.
\begin{description}
 \item[{\rm (1)}] Update $\psi_1$ by adding to it a disjunct from the DNF of 
 $\bigwedge_{e_i\in \ue} \neg 
 \mathtt{ImplDef}_{\psi_1,e_i}^{T_1}(\ux,\uz)$
 and $\psi_2$ by adding to it a disjunct from the DNF of 
 $\bigwedge_{e_i\in \ue} \neg \mathtt{ImplDef}_{\psi_2,e_i}^{T_2}(\ux,\uz)$;
 \item[{\rm (2.i)}] Select $e_i\in \ue$ and $h\in\{1,2\}$; then update $\psi_h$ by adding to it a disjunct $L_{ij}$ from the DNF 
 of $\mathtt{ImplDef}_{\psi_h,e_i}^{T_h}(\ux,\uz)$; the equality $e_i= t_{ij}$ (where $t_{ij}$ is the term mentioned in Lemma~\ref{lem2})\footnote{ 
 Lemma~\ref{lem2} is used taking as $\uy$ the tuple $\ue$, as $\ux$ the tuple $\ux,\uz$, as  $\phi(\ux, \uy)$ the formula $\psi_h(\ux, \uz,\ue)$ and as $\psi$ the formula $\psi_{3-h}$.
 }
 is added to $\mathtt{ExplDef}(\uz, \ux)$; the variable $e_i$ becomes in this way part of the defined variables.  
\end{description}
If alternative (1) is chosen, the procedure stops, otherwise it is recursively applied again and again (we have one truly existential variable less after applying alternative (2.i), so we eventually terminate).
\end{proof}

Thus we are left to the problem of computing a cover of a terminal working formula; this  problem is solved in the following proposition: 
\begin{proposition}\label{prop:terminal}
 A cover of a terminal working formula~\eqref{eq:working} can be obtained just by unravelling the explicit definitions of the  variables $\uz$  from  the formula
 \begin{equation}\label{eq:combinedcover}
 \exists \uz\; (\mathtt{ExplDef}(\uz, \ux) \wedge \theta_1(\ux, \uz) \wedge \theta_2(\ux,\uz))
\end{equation}
where $\theta_1(\ux, \uz)$ is the $T_1$-cover of
 $\exists \ue \psi_1(\ux, \uz, \ue)$ and $\theta_2(\ux,\uz)$ is the $T_2$-cover of
 $\exists \ue \psi_2(\ux, \uz, \ue)$. 
\end{proposition}

\begin{proof}
 In order to show that Formula~\eqref{eq:combinedcover} is the $T_1\cup T_2$-cover of a terminal working formula~\eqref{eq:working}, we prove, by using the Cover-by-Extensions Lemma~\ref{lem:cover}, that, for every $T_1\cup T_2$-model $\cM$, for every tuple $\ua, \uc$ from $\vert \cM\vert$ such that $\cM\models  \theta_1(\ua, \uc) \wedge \theta_2(\ua,\uc)$  there is an extension $\cN$ of $\cM$ such that $\cN$ is still a model of $T_1\cup T_2$ and 
 $\cN\models \exists \ue (\psi_1(\ua, \uc, \ue) \wedge \psi_2(\ua,\uc, \ue))$.  
 By a L\"owenheim-Skolem argument, since our languages are countable, we can suppose that $\cM$ is at most countable and actually that it is countable by stable infiniteness of our theories, see Lemma~\ref{lem:si} (the fact that $T_1\cup T_2$ is stably infinite in case both $T_1, T_2$ are such, comes from the proof of Nelson-Oppen combination result, see~\cite{NO},\cite{TinHar}, \cite{Ghil05}).
 
  According to the conditions~\eqref{eq:ter} and the definition of a cover  (notice that the formulae $\neg\mathtt{ImplDef}_{\psi_h,e_i}^{T_h}(\ux,\uz)$ do not contain the $\ue$ and are quantifier-free) we have that 
 $$
 T_1\vdash \theta_1\to \neg\mathtt{ImplDef}_{\psi_1,e_i}^{T_1}(\ux,\uz) ~~{\rm and}~~T_2\vdash \theta_2\to 
 \neg\mathtt{ImplDef}_{\psi_2,e_i}^{T_2}(\ux,\uz)
$$
(for every $e_i\in \ue$). Thus, since 
  $\cM\not \models  \mathtt{ImplDef}_{\psi_1,e_i}^{T_1}(\ua,\uc)$ and 
 $\cM\not \models  \mathtt{ImplDef}_{\psi_2,e_i}^{T_2}(\ua,\uc)$ 
  holds for every $e_i\in \ue$, 
  we can apply Lemma~\ref{lem1} and conclude that 
 there exist a $T_1$-model $\cN_1$ and a $T_2$-model $\cN_2$ such that $\cN_1\models \psi_1(\ua, \uc, \ub_1)$ and 
 $\cN_2\models \psi_2(\ua, \uc, \ub_2)$ for tuples $\ub_1\in \vert \cN_1\vert$ and $\ub_2\in \vert \cN_2\vert$, both disjoint from $\vert\cM\vert$.  By a L\"owenheim-Skolem argument,  we can suppose that $\cN_1, \cN_2$ are countable and by 
 Lemma~\ref{lem:si} even that they are both countable extensions of $\cM$.
 
 The tuples $\ub_1$ and $\ub_2$ have equal length because the $\psi_1, \psi_2$ from our working formulae entail  $e_i \neq e_j$, where $e_i, e_j$ are different existential variables. 
 Thus there is a bijection $\iota: \vert \cN_1\vert \to \vert \cN_2\vert$ fixing all elements in $\cM$ and mapping component-wise the $\ub_1$ onto the $\ub_2$. But this means that, exactly 
 as it happens in the proof of the completeness of the Nelson-Oppen combination procedure,  the $\Sigma_2$-structure on $\cN_2$ can be moved back via $\iota^{-1}$ to $\vert \cN_1\vert$ in such a way that the $\Sigma_2$-substructure from $\cM$ is fixed and in such a way that the tuple $\ub_2$ is mapped to the tuple $\ub_1$. In this way, $\cN_1$ becomes a $\Sigma_1\cup\Sigma_2$-structure which is a model of $T_1\cup T_2$ and which is such that $\cN_1\models \psi_1(\ua, \uc, \ub_1) \wedge \psi_2(\ua,\uc, \ub_1)$, as required.
\end{proof}

From Lemma~\ref{lem:toterminal}, Proposition~\ref{prop:terminal} and 
Theorem~\ref{prop:qe}, we immediately get

\begin{theorem}\label{thm:main} Let $T_1, T_2$ be convex, stably infinite, equality interpolating, universal theories over disjoint signatures admitting a model completion.  Then $T_1\cup T_2$ admits a model completion too. Covers in $T_1\cup T_2$ can be effectively computed as shown above.
\end{theorem}

Notice that the input cover algorithms in the above combined cover computation algorithm are used not only in the final step described in Proposition~\ref{prop:terminal}, but also every time we need to compute 
a formula $\mathtt{ImplDef}_{\psi_h,e_i}^{T_h}(\ux,\uz)$: according to its definition, this formula is obtained by eliminating quantifiers in $T_i^*$ from~\eqref{eq:def} (this is done via a cover computation, reading $\forall$ as $\neg \exists\neg$). In practice, implicit
definability is not very frequent,
so that in many concrete cases
$\mathtt{ImplDef}_{\psi_h,e_i}^{T_h}(\ux,\uz)$ is trivially equivalent to $\bot$ (in such cases, Step (2.i) above can obviously be disregarded). 

%
%
%

 \subsubsection{An Example.}\label{sec:example}

We now analyze an example in detail. 
Our results apply for instance to the case where $T_1$ is $\EUF(\Sigma)$ 
and $T_2$ is linear real arithmetic. We recall that covers are computed  in real arithmetic by quantifier elimination, whereas for $\EUF(\Sigma)$ one can apply the superposition-based algorithm from~\cite{cade19}.
 Let us show that  the cover of
 \begin{equation}\label{eq:ex0}
 \exists e_1\cdots \exists e_4\;
~ \left( 
\begin{aligned}
&
e_1= f(x_1)\;\wedge\; e_2= f(x_2)\; \wedge\; 
\\ & \wedge\,
f(e_3)=e_3\,\wedge\, \; f(e_4)= x_1\;\wedge
\\ & \wedge\;
x_1+e_1\leq e_3\;\wedge\; e_3\leq x_2+e_2\;\wedge \;e_4=x_2+e_3
 \end{aligned}
 \right)
\end{equation}
is the following formula
\begin{equation}\label{eq:res}
\begin{aligned}
&
[x_2=0 \; \wedge \; f(x_1)=x_1 \;\wedge \; x_1 \leq 0 \; \wedge\; x_1\leq f(0)]~\vee~
\\ & \vee ~
[x_1+f(x_1)< x_2+f(x_2)\,\wedge \,x_2\neq 0]~\vee
\\ & \vee ~
\begin{bmatrix} x_2\neq 0\, \wedge\, x_1+f(x_1)= x_2+f(x_2)\, \wedge\, f(2x_2+f(x_2))=x_1\, \wedge\,
\\   ~ \wedge\, f(x_1+f(x_1))= x_1+f(x_1)\end{bmatrix}
\end{aligned}
\end{equation}
%

Formula~\eqref{eq:ex0} is already purified.
 Notice also that the variables $e_1, e_2$ are in fact already explicitly  defined (only $e_3, e_4$ are truly existential variables). 
 
 We first make the partition guessing. There is no need to involve defined variables into the partition guessing, hence we need to consider only two partitions; they  are described by the following  formulae:
 \begin{eqnarray*}
  P_1(e_3, e_4) ~&\equiv&~ e_3\neq e_4 
  \\
  P_2(e_3, e_4) ~&\equiv&~ e_3= e_4 
 \end{eqnarray*}
We first analyze  \textbf{the case of $P_1$}.  The formulae $\psi_1$ and $\psi_2$ to which we need to apply exhaustively Step (1) and Step (2.i) of our algorithm are:
\begin{eqnarray*}
 \psi_1 ~&\equiv&~f(e_3)=e_3\,\wedge\, \; f(e_4)= x_1\;\wedge\; e_3\neq e_4
 \\
 \psi_2 ~&\equiv&~x_1+e_1\leq e_3\;\wedge\; e_3\leq x_2+e_2\;\wedge \;e_4=x_2+e_3\;
 \wedge\; e_3\neq e_4
\end{eqnarray*}
 We first compute the implicit definability formulae for the truly existential variables with respect to both $T_1$ and $T_2$.
\vspace{-1mm}
\begin{description}
 \item[-] We first consider $\mathtt{ImplDef}_{\psi_1,e_3}^{T_1}(\ux,\uz)$.  Here we show that the cover of the negation of formula~\eqref{eq:def} is equivalent to $\top$ 
 (so that $\mathtt{ImplDef}_{\psi_1,e_3}^{T_1}(\ux,\uz)$ is equivalent to $\bot$).
 We must quantify over truly existential variables and their duplications, thus we need to compute the cover of
 $$
 f(e'_3)=e'_3\wedge f(e_3)=e_3\wedge  f(e'_4)=x_1\wedge f(e_4)=x_1\wedge e_3\neq e_4 \wedge e'_3 \neq e'_4\wedge  e'_3\neq e_3
 $$
  This is a saturated set according to the superposition based procedure of~\cite{cade19}, hence the result is $\top$, as claimed.
 \item[-] The formula $\mathtt{ImplDef}_{\psi_1,e_4}^{T_1}(\ux,\uz)$ is also equivalent to $\bot$, by the same  argument as above.
 \item[-] To compute $\mathtt{ImplDef}_{\psi_2,e_3}^{T_2}(\ux,\uz)$ we  use Fourier-Motzkin quantifier elimination. We need to eliminate the variables $e_3, e'_3, e_4, e'_4$ (intended as existentially quantified variables) from 
 \begin{equation*}
 \begin{aligned}
 &
 x_1+e_1\leq e'_3\leq x_2+e_2\wedge x_1+e_1\leq e_3\leq x_2+e_2\wedge e'_4=x_2+e'_3\wedge
 \\ & \wedge
 e_4=x_2+e_3 \wedge e_3\neq e_4 \wedge e'_3 \neq e'_4\wedge e'_3\neq e_3~~.
 \end{aligned}
 \end{equation*}
 This gives $x_1+e_1\neq x_2+e_2 \wedge x_2\neq 0$, so that $\mathtt{ImplDef}_{\psi_2,e_3}^{T_2}(\ux,\uz)$ is 
 $x_1+e_1= x_2+e_2 \wedge x_2\neq 0$. The corresponding equality interpolating term for $e_3$ is $x_1+e_1$. 
 \item[-] The formula $\mathtt{ImplDef}_{\psi_2,e_4}^{T_2}(\ux,\uz)$ is also equivalent to 
 $x_1+e_1= x_2+e_2\wedge x_2\neq 0$ and the  equality interpolating term for $e_4$ is $x_1+e_1+x_2$.
\end{description}
 
So, if we apply Step 1 we get
 \begin{equation}\label{eq:ex1}
 \exists e_1\cdots \exists e_4
 \left( 
\begin{aligned}
&
e_1= f(x_1)\;\wedge\; e_2=f(x_2)\; \wedge
\\ & \wedge\,
f(e_3)=e_3\,\wedge\, \; f(e_4)= x_1\;\wedge \; e_3\neq e_4\;\wedge
\\ & \wedge\,
 x_1+e_1\leq e_3\,\wedge\, e_3\leq x_2+e_2\,\wedge \, e_4=x_2+e_3
 \, \wedge \, x_1+e_1\neq x_2+e_2
 \end{aligned}
 \right)
\end{equation}
(notice that the literal $x_2\neq 0$ is entailed by $\psi_2$, so we can simplify it to $\top$ in $\mathtt{ImplDef}_{\psi_2,e_3}^{T_2}(\ux,\uz)$ and $\mathtt{ImplDef}_{\psi_2,e_4}^{T_2}(\ux,\uz)$).
 If we apply Step
 (2.i) (for i=3), we get (after removing implied equalities)
 \begin{equation}\label{eq:ex2}
 \exists e_1\cdots \exists e_4\;
~ \left( 
\begin{aligned}
&
e_1=f(x_1)\;\wedge\; e_2=f(x_2)\; \wedge\; e_3=x_1+e_1 \;\wedge
\\ & \wedge\,
f(e_3)=e_3\,\wedge\, \; f(e_4)= x_1\;\wedge\; e_3\neq e_4\;\wedge
\\ & \wedge\;
e_4=x_2+e_3 \;\wedge\;
 x_1+e_1= x_2+e_2
 \end{aligned}
 \right)
\end{equation}
Step (2.i) (for i=4)  gives a formula logically equivalent to~\eqref{eq:ex2}.
Notice that~\eqref{eq:ex2} is terminal too, because all existential variables are now explicitly defined (this is a 
lucky
side-effect of the fact that $e_3$ has been moved to the defined variables).
Thus the exhaustive application of Steps (1) and (2.i) is concluded.

Applying the final step of Proposition~\ref{prop:terminal} to~\eqref{eq:ex2} is quite easy: it is sufficient to unravel the acyclic definitions. The result, after little simplification, is
\begin{equation*}
\begin{aligned}
& x_2\neq 0\, \wedge\, x_1+f(x_1)= x_2+f(x_2)\,  \wedge\, \\ & \wedge\, f(x_2+f(x_1+f(x_1)))=x_1\, \wedge\, f(x_1+f(x_1))= x_1+f(x_1);
 \end{aligned}
\end{equation*}
this can be further simplified to
\begin{equation}\label{eq:res11}
\begin{aligned}
 & x_2\neq 0\, \wedge\, x_1+f(x_1)= x_2+f(x_2)\, \wedge \\ & \wedge\, f(2x_2+f(x_2))=x_1\, \wedge\, f(x_1+f(x_1))= x_1+f(x_1);
 \end{aligned}
\end{equation}

As to formula~\eqref{eq:ex1}, we need to apply  the final cover computations mentioned in Proposition~\ref{prop:terminal}. The formulae $\psi_1$ and $\psi_2$ are now
\begin{eqnarray*}
 \psi'_1~\equiv~
 &
 f(e_3)=e_3\;\wedge\, \; f(e_4)= x_1\;\wedge\, \; e_3\neq e_4~~~~~~~~~~~~~~~~
 \\
 \psi'_2~\equiv~
 &
 x_1+e_1\leq e_3 \leq x_2+e_2\,\wedge \, e_4=x_2+e_3
 \, \wedge \, x_1+e_1\neq x_2+e_2\;\wedge \; e_3\neq e_4
\end{eqnarray*}
The $T_1$-cover of $\psi_1'$ is $\top$. For the $T_2$-cover of $\psi_2'$,  
eliminating with Fourier-Motzkin the variables $e_4$ and  $e_3$, we get
\begin{equation*}
x_1+e_1< x_2+e_2\,\wedge \,x_2\neq 0
\end{equation*}
which becomes
\begin{equation}\label{eq:res12}
x_1+f(x_1)< x_2+f(x_2)\,\wedge \,x_2\neq 0
\end{equation}
after unravelling the explicit definitions of $e_1, e_2$.
Thus, \emph{the analysis of the case of the partition $P_1$ gives, as a result, the disjunction of~\eqref{eq:res11} and 
~\eqref{eq:res12}.}

We now analyze  \textbf{the case of $P_2$}. Before proceeding, we replace $e_4$ with $e_3$ (since $P_2$ precisely asserts that these two variables coincide); our formulae $\psi_1$ and $\psi_2$ become
\begin{eqnarray*}
 \psi''_1 ~&\equiv&~f(e_3)=e_3\,\wedge\, \; f(e_3)= x_1
 \\
 \psi''_2 ~&\equiv&~x_1+e_1\leq e_3\;\wedge\; e_3\leq x_2+e_2\;\wedge \;0=x_2\;
\end{eqnarray*}
From $\psi''_1$ we deduce $e_3=x_1$, thus we can move $e_3$ to the explicitly defined variables (this avoids useless calculations:  the implicit definability condition for  variables  having an entailed explicit definition is obviously $\top$, so making case split on it produces either tautological consequences or inconsistencies). In this way we get the terminal working formula 
\begin{equation}\label{eq:ex3}
 \exists e_1\cdots \exists e_3
 \left( 
\begin{aligned}
&
e_1= f(x_1)\;\wedge\; e_2=f(x_2)\; \wedge \; e_3=x_1
\\ & \wedge\,
f(e_3)=e_3\,\wedge\, \; f(e_3)= x_1\;\wedge \; 
\\ & \wedge\,
 x_1+e_1\leq e_3\,\wedge\, e_3\leq x_2+e_2\,\wedge \, 0=x_2
 \end{aligned}
 \right)
\end{equation}
Unravelling the explicit definitions, we get (after exhaustive simplifications)
\begin{equation}\label{eq:res2}
 x_2=0 \; \wedge \; f(x_1)=x_1 \;\wedge \; x_1 \leq 0 \; \wedge\; x_1\leq f(0)  
\end{equation}

Now, the disjunction of~\eqref{eq:res11},\eqref{eq:res12} and~\eqref{eq:res2} is precisely the final result~\eqref{eq:res} claimed above.
This concludes our detailed analysis of our example. 

Notice that the example shows that combined cover computations may introduce 
terms with arbitrary alternations of symbols from both theories
 (like 
 $f(x_2+f(x_1+f(x_1)))$ above). The point is that when a variable becomes explicitly definable via a term in one of the theories, then using such additional variable may in turn cause some other variables  to become explicitly definable via terms from the other theory, and so on and so forth; when ultimately the explicit definitions are unraveled, highly nested terms arise with many symbol alternations from both theories.

\noindent\subsubsection{The Necessity of the Equality Interpolating Condition.}\label{sec:necessity}

The following  result shows that  equality interpolating is   a necessary condition for a transfer result, in the sense that it is already required for minimal combinations with signatures adding uninterpreted symbols:



\begin{theorem}\label{thm:necessary} Let $T$ be a convex, stably infinite,  universal theory admitting a model completion
and let $\Sigma$ be a signature disjoint from the signature of $T$ containing at least a unary predicate symbol.
Then $T\cup \EUF(\Sigma)$ admits a model completion iff $T$ is equality interpolating.
\end{theorem}

\begin{proof} 
The necessity can be shown by using the following argument. 
By Theorem~\ref{prop:qe}, 
 $T\cup \EUF(\Sigma)$ has uniform quantifier-free interpolation, hence also ordinary quantifier-free interpolation.
  We can now apply Theorem~\ref{prop:strong_amalgamation_needed} and get that $T$ must be equality interpolating. 
 Conversely, the sufficiency comes from Theorem~\ref{thm:main} together with the fact that
 $\EUF(\Sigma)$ is trivially universal, convex, stably infinite,
 has a model completion~\cite{cade19}
 and is equality interpolating~\cite{ym},\cite{bgr-acmtocl}.
 \end{proof}

\section{The Non-Convex Case: a Counterexample}\label{sec:nonconvex}

In this section, we show  by giving a suitable counterexample that the convexity hypothesis cannot be dropped from Theorems~\ref{thm:main},~\ref{thm:necessary}. We make use of basic facts about ultrapowers (see~\cite{CK} for the essential information we need).
We take as $T_1$ integer difference logic \IDL, i.e. the theory of integer numbers under the unary operations of successor and predecessor, the constant 0 and the strict order  relation $<$. This is stably infinite,
universal and has quantifier elimination (thus it coincides with its own model completion). 
It is not convex, but it satisfies the equality interpolating condition, once the latter is suitably adjusted to non-convex theories, see~\cite{bgr-acmtocl}  for the related definition and all the above mentioned facts.

As $T_2$, we take $\EUF(\Sigma_{f})$, where $\Sigma_f$  has just one unary free function symbol $f$ (this $f$ is supposed not to belong to the signature of $T_1$).

\begin{proposition}
 Let $T_1, T_2$ be as above; the formula 
 \begin{equation}\label{eq:cex}
  \exists e\;(0<e \wedge e< x\wedge f(e)=0)
 \end{equation}
does not have a cover in $T_1\cup T_2$.
\end{proposition}

\begin{proof}
 Suppose that~\eqref{eq:cex} has a cover $\phi(x)$. This means (according to Cover-by-Extensions Lemma~\ref{lem:cover})  that for every model $\cM$ of $T_1\cup T_2$ and for every element $a\in \vert \cM\vert$ such that $\cM\models\phi(a)$, there is an extension $\cN$ of $\cM$ such that $\cN\models
 \exists e\;(0<e \wedge e< a\wedge f(e)=0)$.
 
 Consider the model $\cM$,
 so specified: the support of $\cM$ is the set of the integers, the symbols from the signature of $T_1$ are interpreted in the standard way and the symbol $f$ is interpreted so that 0 is not in the image of $f$. 
 Let $a_k$ be the number $k>0$ (it is an element from the support of $\cM$).
 Clearly it is not possible to 
 extend $\cM$ so that $\exists e\;(0<e \wedge e< a_k\wedge f(e)=0)$ becomes true: 
 indeed, we know that all the elements in the interval $(0,k)$ are definable as iterated successors of $0$ and, by using the axioms of \IDL, no element can be added between a number and its successor, hence this interval cannot be enlarged in a superstructure. We conclude that $\cM  \models \neg\phi(a_k)$ for every $k$.  
 
 Consider now an ultrapower $\Pi_D \cM$ of $\cM$ modulo a non-principal ultrafilter $D$ and let $a$ be the equivalence class of the tuple 
 $\langle a_k\rangle_{k\in \mathbb{N}}$; by the fundamental L{os} theorem~\cite{CK},  $\Pi_D \cM\models \neg\phi(a)$. 
We claim that it is possible to extend $\Pi_D \mathcal{M}$ to a superstructure $\mathcal N$ such that $\cN\models
 \exists e\;(0<e \wedge e< a\wedge f(e)=0)$:  this would entail, by definition of cover, that  $\Pi_D \mathcal{M}\models \phi(a)$, contradiction.
 We now show why the claim is true. Indeed, since  $\langle a_k\rangle_{k\in\mathbb{N}}$ has arbitrarily big numbers as its components, we have that, in $\Pi_D \cM$,  $a$ is bigger than all standard numbers. 
  %
  Thus, if we take a further non-principal ultrapower $\cN$ of  $\Pi_D \cM$, it becomes possible to change in it the evaluation of $f(b)$ for some $b<a$ and set it to $0$ (in fact, as it can be easily seen, 
  there are 
  elements  $b\in  \vert \cN\vert$ less than $a$ but not in the support of $\Pi_D \cM$). 
\end{proof}

The counterexample still applies  when replacing integer difference logic with linear integer arithmetics.

\vspace{-2mm}

\section{Tame Combinations}\label{sec:value}

So far, we only analyzed the mono-sorted case. However, many interesting examples arising in model-checking verification are multi-sorted: this is the case of array-based systems~\cite{lmcs} and in particular of the array-based system used in data-aware verification~\cite{CGGMR19},\cite{MSCS20}.
The above examples suggest restrictions on the theories to be combined other than convexity, in particular they suggest restrictions that make sense in a multi-sorted context.

Most definitions we gave in Section~\ref{sec:prelim} have straightforward natural extensions to the multi-sorted case (we leave the reader to formulate them). A little care is needed however for the disjoint signatures requirement.
Let $T_1, T_2$ be multisorted theories in the signatures $\Sigma_1, \Sigma_2$; the disjointness requirement for 
$\Sigma_1$ and $\Sigma_2$ can be formulated in this context by saying that  the only function or relation symbols in $\Sigma_1\cap \Sigma_2$ are the equality predicates over the common sorts in $\Sigma_1\cap \Sigma_2$. We want to strengthen this requirement: we say that the combination $T_1\cup T_2$ is \emph{tame} iff the sorts in $\Sigma_1\cap \Sigma_2$ \emph{can only be the codomain sort} (and not a domain sort) of a symbol from $\Sigma_1$ other than an equality predicate. In other word, if a relation or a function symbol has as among its domain sorts a sort from $\Sigma_1\cap \Sigma_2$, then this symbol is from $\Sigma_2$ (and not from $\Sigma_1$, unless it is the equality predicate).

Tame combinations arise in infinite-state model-checking (in fact, the definition is suggested by this application domain), where  signatures can be split into a  signature $\Sigma_2$ for 'data' and a signature $\Sigma_1$ for 'data containers', see~\cite{CGGMR19},\cite{MSCS20}. 

Notice that the notion of a tame combination is not symmetric in $T_1$ and $T_2$: to see this, notice that if the sorts of $\Sigma_1$ are included in the sorts of $\Sigma_2$, then $T_1$ must be a pure equality theory (but this  is not the case if we swap $T_1$ with $T_2$). The combination of $\IDL$ and $\EUF(\Sigma)$ used in  the 
counterexample of section~\ref{sec:nonconvex} is not tame: even if we formulate $\EUF(\Sigma)$ as a two-sorted theory, the unique sort of $\IDL$ must  be a sort of $\EUF(\Sigma)$ too, as witnessed by the impure atom $f(e)=0$ in the formula~\eqref{eq:cex}. Because of this, for the combination to be tame, $\IDL$ should play the role of $T_2$ (the arithmetic operation symbols are defined on a  shared sort); however, the unary function symbol $f\in \Sigma$ has a shared sort as domain sort, so the combination is not tame anyway.

In a tame combination, an atomic formula $A$ can only be of two kinds: (1) we say that $A$ is of the \emph{first kind} iff the sorts of its root predicate are from $\Sigma_1\setminus \Sigma_2$; (2) we say that $A$ 
is of the \emph{second kind} iff the sorts of its root predicate are from $\Sigma_2$.
We use the roman letters $e, x,\dots$ for variables ranging over sorts in $\Sigma_1\setminus \Sigma_2$ and the greek letters $\eta, \xi, \dots$ for variables ranging over sorts in $\Sigma_2$. Thus, if we want to display free variables,  atoms of the first kind can be represented as $A(e,x, \dots)$, whereas atoms
of the second kind can be represented as $A(\eta, \xi, \dots, t(e, x, \dots), \dots)$, where the $t$ are $\Sigma_1$-terms. 

Suppose not that $T_1\cup T_2$ is a tame combination and that $T_1, T_2$ are universal theories admitting  model completions $T_1^*, T_2^*$. We propose the following algorithm to compute the cover of a primitive formula; the latter  must be of the kind 
\begin{equation}\label{eq:tamecover}
 \exists \ue\;\exists \ueta (\phi(\ue, \ux) \wedge \psi(\ueta, \uxi, \ut(\ue,\ux)))
\end{equation}
where $\phi$ is a $\Sigma_1$-conjunction of literals, $\psi$ is a conjunction of $\Sigma_2$-literals and the $\ut$ are $\Sigma_1$-terms.
The algorithm has three steps.
\begin{description}
 \item[{\rm (i)}] We flatten~\eqref{eq:tamecover} and get
 \begin{equation}\label{eq:tamecover1}
 \exists \ue\;\exists \ueta\; \exists \ueta'\; (\phi(\ue, \ux) \wedge \ueta'= \ut(\ue, \ux)\wedge\psi(\ueta, \uxi, \ueta')))
\end{equation}
where the $\ueta'$ are fresh variables  abstracting out the $\ut$ and $\ueta'= \ut(\ue, \ux)$  is a componentwise conjunction of equalities.
\item[{\rm (ii)}] We apply the cover algorithm of $T_1$ to the formula
\begin{equation}\label{eq:tamecover1+}
 \exists \ue\; (\phi(\ue, \ux) \wedge \ueta'= \ut(\ue, \ux))~~;
\end{equation}
this gives as a result a formula $\tilde\phi(\ux, \ueta')$ that we put in DNF. A disjunct of $\phi$ will have the form 
$\phi_1(\ux)\wedge \phi_2(\ueta', \ut'(\ux))$ after separation of the literals of the first and of the second kind. 
We pick such a disjunct 
$\phi_1(\ux)\wedge \phi_2(\ueta', \ut'(\ux))$ of the DNF of $\tilde\phi(\ux, \ueta')$ and update our current primitive formula to
 \begin{equation}\label{eq:tamecover2}
 \exists \uxi'\;( \uxi'= \ut'(\ux)\wedge (\exists \ueta\; \exists \ueta'\; (\phi_1(\ux)\wedge \phi_2(\ueta', \uxi') \wedge \psi(\ueta, \uxi, \ueta'))))
\end{equation}
(this step is nondeterministic: in the end we shall output the disjunction of all possible outcomes). Here again the $\uxi'$ are fresh variables abstracting out the terms $\ut'$. Notice that, according to the definition of a tame combination, $\phi_2(\ueta', \uxi')$ must be a conjunction of equalities and disequalities between variable terms, because it is a $\Sigma_1$-formula (it comes from a $T_1$-cover computation) and $\ueta', \uxi'$ are variables of $\Sigma_2$-sorts.
\item[{\rm (iii)}] We apply the cover algorithm of $T_2$ to the formula
 \begin{equation}\label{eq:tamecover2+}
 \exists \ueta\; \exists \ueta'\; (\phi_2(\ueta', \uxi') \wedge \psi(\ueta, \uxi, \ueta'))
\end{equation}
this gives as a result a formula $\psi_1(\uxi, \uxi')$. We update our current formula to 
$$\exists \uxi'\;( \uxi'= \ut'(\ux)\wedge \phi_1(\ux)\wedge \psi_1( \uxi, \uxi' ))$$
and finally to the equivalent  quantifier-free formula 
\begin{equation}\label{eq:tamecover2++}
 \phi_1(\ux)\wedge \psi_1( \uxi, \ut'(\ux) )~~.
\end{equation}
\end{description}

\noindent
We now show that the above algorithm is correct under very mild hypotheses. We need some technical facts about stably infinite theories in a multi-sorted context. We say that a multi-sorted theory $T$ is \emph{stably infinite with respect to a set of sorts $\cS$ from its signature} iff every $T$-satisfiable constraint is satisfiable in a model $\cM$ where, for every $S\in \cS$, the set  $S^\cM$ (namely the interpretation of the sort $S$ in $\cM$) is  
infinite. The next Lemma is a light generalization of Lemma~\ref{lem:si} and is proved in the same way (the proof is reported in Appendix~\ref{app:lemma6}):

\begin{lemma}\label{lemma6}
 Let $T$ be stably infinite with respect to a subset $\cS$ of the set of sorts of the signature of $T$. Let $\cM$ be a model of $T$ and  let, for every $S\in \cS$,  $X_S$ be a 
 superset of  $S^\cM$.
 Then there is an extension $\cN$ of $\cM$  such that for all $S\in \cS$ we have $S^\cN\supseteq X_S$. 
\end{lemma}

\begin{lemma}\label{lem:stable}
 Let $T_1, T_2$ be universal signature disjoint theories which are stably infinite with respect to the set of shared sorts 
(we let $\Sigma_1$ be the signature of $T_1$ and $\Sigma_2$ be the signature of $T_2$).
 Let $\cM_0$ be  model of $T_1\cup T_2$ and let  $\cM_1$ be a model of $T_1$ extending the $\Sigma_1$-reduct of $\cM_0$. Then there exists a
 model $\cN$ of $T_1\cup T_2$, extending $\cM_0$ as a $\Sigma_1\cup \Sigma_2$-structure and whose $\Sigma_1$-reduct  extends $\cM_1$.
\end{lemma}

\begin{proof} 
Using the previous lemma,
 build a chain of models $\cM_0\subseteq \cM_1\subseteq \cM_2\subseteq\cdots$ such that for all $i$,  $\cM_{2i}$ is a model of $T_2$,  $\cM_{2i+1}$ is a model of $T_1$ and $\cM_{2i+2}$ is a $\Sigma_2$-extension of $\cM_{2i}$, whereas $\cM_{2i+3}$ is a $\Sigma_2$-extension of $\cM_{2i+1}$. The union over this chain of models will be the desired $\cN$.
\end{proof}

We are now ready for the main result of this section:

\begin{theorem}\label{thm;main1}
 Let $T_1\cup T_2$ be a tame combination of two universal theories admitting a model completion. If $T_1,T_2$ are also stably infinite
 with repect to their shared sorts, then $T_1\cup T_2$ has a model completion. Covers in $T_1\cup T_2$ can be computed as shown in the above three-steps algorithm.
\end{theorem}
\begin{proof}
Since condition (i) of Lemma~\ref{lem:cover} is trivially true, we need only to check condition (ii), namely that given a $T_1\cup T_2$-model $\cM$ and elements $\ua, \ub$ from its support such that $\cM\models \phi_1(\ua)\wedge \psi_1( \ub, \ut'(\ua) )$ as in~\eqref{eq:tamecover2++}, then there is an extension $\cN$ of $\cM$ such that~\eqref{eq:tamecover} is true in $\cN$ when evaluating $\ux$ over $\ua$ and $\uxi$ over $\ub$.

If we let $\ub'$ be the tuple such that $\cM\models \ub'= \ut'(\ua)$, then we have $\cM\models \ub'= \ut'(\ua)\wedge \phi_1(\ua)\wedge \psi_1( \ub, \ub' )$. Since $\psi_1(\uxi, \uxi')$ is the $T_2$-cover of~\eqref{eq:tamecover2+}, the $\Sigma_2$-reduct of $\cM$ embeds into a $T_2$-model
where~\eqref{eq:tamecover2+} is true under the evaluation of the $\uxi$  as the $\ub$. By Lemma~\ref{lem:stable},
this model can be embedded into a $T_1\cup T_2$-model $\cM'$ in such a way that $\cM'$ is an extension of $\cM$ and that $\cM'\models  \ub'= \ut'(\ua)\wedge \phi_1(\ua)\wedge \phi_2(\uc', \ub') \wedge \psi(\uc, \ub, \uc')$ for some $\uc, \uc'$.
Since $\phi_1(\ux)\wedge \phi_2(\ueta', \ut'(\ux))$ implies the $T_1$-cover of~\eqref{eq:tamecover1+} and 
$\cM'\models \phi_1(\ua)\wedge \phi_2(\uc', \ut(\ua))$, then the $\Sigma_1$-reduct of $\cM'$ can be expanded to a $T_1$-model where~\eqref{eq:tamecover1+} is true when evaluating the $\ux, \ueta'$ to the $\ua, \uc'$. Again by Lemma~\ref{lem:stable}, this model can be expanded to a $T_1\cup T_2$-model $\cN$ such that $\cN$ is an extension of $\cM'$ (hence also of $\cM$) and 
$\cN\models \phi(\ua', \ua) \wedge \uc'= \ut(\ua', \ua)\wedge \psi(\uc, \ub, \uc') $, that is 
$\cN\models \phi(\ua', \ua) \wedge \psi(\uc, \ub, \ut(\ua', \ua)) $. This means that $\cN\models 
\exists \ue\;\exists \ueta (\phi(\ue, \ua) \wedge \psi(\ueta, \ub,\ut(\ue,\ua))) $, as desired.
\end{proof}

\section{Conclusions and Future Work}\label{sec:conclusions}


In this paper we showed that covers (aka uniform interpolants) exist in the combination of two convex universal  theories over disjoint signatures in case they exist in the component theories and in case the component theories also satisfy the equality interpolating condition - this further condition 
is nevertheless needed in order to transfer the existence of (ordinary) quantifier-free interpolants. 
In order to prove that, Beth definability property for primitive fragments turned out to be the crucial ingredient to extensively employ. In case convexity fails, we showed by a counterexample that covers might not exist anymore in the combined theory.  
%
The last result raises the following research problem.
Even if in general covers do not exist for combination of non-convex theories, it would be interesting to see under what conditions one can decide whether a given cover exists and, in the affirmative case, to compute it.

Applications suggested a different line of investigations, i.e., what we called `tame combinations'. 
In database-driven verification~\cite{CGGMR19},\cite{BPM19},\cite{MSCS20} one uses tame combinations $T_1\cup T_2$, where $T_1$ is a multi-sorted version of $\EUF(\Sigma)$ in a signature $\Sigma$ containing only unary function symbols and relations of any arity. In this context, quantifier elimination in $T_1^*$ for primitive formulae is quadratic in complexity. Model-checkers like \textsc{MCMT} represent sets of reachable states by using conjunctions of literals and primitive formulae to which quantifier elimination should be applied arise from preimage computations. Now, in this context, if all relation symbols are at most binary, then quantifier elimination in $T_1^*$ produces conjunctions of literals out of primitive formulae, thus step (ii) in the above algorithm becomes deterministic and the only reason why the  algorithm may become expensive (i.e. non polynomial) lies in the final quantifier elimination step for $T_2^*$. The latter might be extremely expensive if substantial arithmetic is involved, but it might still be efficiently handled in practical cases where only very limited arithmetic (e.g. difference bound constraints like $x-y\leq n$ or $x\leq n$, where $n$ is a constant) is involved.  This is why we feel that our algorithm for covers in tame combinations can be really useful in the applications.
This is confirmed by our first experiments with version 2.9 of \textsc{MCMT}, where such algorithm has been implemented.

A final future research line could consider cover transfer properties to non-disjoint signatures combinations, analogously to similar results obtained in~\cite{GG18} for the transfer of quantifier-free interpolation. Indeed, the main challenge here seems to consist in finding sufficient condition for existence of covers in combination of non-convex theories: in fact, we know from Section~\ref{sec:nonconvex} that the non-convex version of the equality interpolation property~\cite{bgr-acmtocl} is not enough for this purpose.

\bibliographystyle{abbrv}
\bibliography{mcmt}

\newpage

\appendix

\section{Appendix}\label{sec:approofs}

In this Appendix we report the proof of the Cover-by-Extensions Lemma~\ref{lem:cover}, 
 of the technical Lemmas~\ref{lem:si} and~\ref{lem1} and fill the missing details of the proof of Lemma~\ref{lem:toterminal}.

\subsection{Proof of Lemma~\ref{lem:cover}}

The Cover-by-Extension Lemma is not an original result of this paper: the proof is reported here from~\cite{cade19} just for the sake of completeness (the  Lemma is crucial in the present paper too).

\vskip 2mm\noindent
\textbf{Lemma \ref{lem:cover} [Cover-by-Extensions]} \emph{A formula $\psi(\uy)$ is a $T$-cover of $\exists \ue\, \phi(\ue, \uy)$ iff 
it satisfies the following two conditions:
\begin{inparaenum}[(i)]
\item $T\models  \forall \uy\,( \exists \ue\,\phi(\ue, \uy) \to \psi(\uy))$;
\item for every model $\cM$ of $T$, for every tuple of  elements $\ua$ from the support of $\cM$ such that $\cM\models \psi(\ua)$ it is possible to find
  another model $\cN$ of $T$ such that $\cM$ embeds into $\cN$ and $\cN\models \exists \ue \,\phi(\ue, \ua)$.
\end{inparaenum}
}
\vskip 1mm
\begin{proof}
\textit{Suppose that $\psi(\uy)$ satisfies conditions (i) and (ii) above}. 
Condition (i) says that $\psi(\uy)\in Res(\exists \ue\, \phi)$, so $\psi$ is a residue.
 In order to show that $\psi$ is also a cover, we have to prove that $T\models \forall \uy,\uz (\psi(\uy)\to \theta(\uy,\uz))$, for every $\theta(\uy,\uz)$ that is a residue for $\exists \ue\, \phi(\ue, \uy)$.
 Given a model $\cM$ of $T$, take a pair of tuples $\ua, \ub$ of elements from $|\cM|$ and suppose that $\cM\models \psi(\ua)$. By condition (ii), there is a model $\cN$ of $T$ such that $\cM$ embeds into $\cN$ and $\cN\models \exists \ue \phi(\ue, \ua)$. Using the definition of $Res(\exists \ue\, \phi)$, we have $\cN\models \theta(\ua,\ub)$, since $\theta(\uy,\uz)\in Res(\exists \ux\, \phi)$. Since $\cM$ is a substructure of $\cN$ and $\theta$ is quantifier-free, $\cM\models \theta(\ua,\ub)$ as well, as required.

\textit{Suppose that $\psi(\uy)$ is a cover}. The definition of residue implies condition (i). To show condition (ii) we have to prove that, given a model $\cM$ of $T$, for every tuple $\ua$ of elements from $|\cM|$, if $\cM\models \psi(\ua)$, then there exists a model $\cN$ of $T$ such that $\cM$ embeds into $\cN$ and $\cN\models \exists \ue \phi(\ue, \ua)$. By reduction to absurdity, suppose that this is not the case: this is equivalent (by using Robinson Diagram Lemma) to the fact that $\Delta(\cM)\cup \{ \phi(\ue, \ua) \}$ is a $T$-inconsistent $\Sigma^{|\cM|\cup\{\ue\}}$-theory. By compactness, there is a finite number of literals $\ell_1(\ua,\ub),...,\ell_m(\ua,\ub)$ (for some tuple $\ub$ of elements from $|\cM|$) such that $\cM\models \ell_i$ (for all $i=1,\dots,m$) and  $T\models\phi(\ue, \ua)\to \neg (\ell_1(\ua,\ub)\land\cdots\land \ell_m(\ua, \ub))$, which means that 
        $T\models \phi(\ue, \uy)\to (\neg \ell_1(\uy,\uz)\lor \cdots \lor \neg \ell_m(\uy,\uz))$, i.e. that 
        $T\models \exists\ue\,\phi(\ue, \uy)\to (\neg \ell_1(\uy,\uz)\lor \dots \lor \neg \ell_m(\uy,\uz))$. By definition of residue, clearly $(\neg \ell_1(\uy,\uz)\lor \dots \lor \neg \ell_m(\uy,\uz)) \in Res(\exists \ux\, \phi)$; then, since $\psi(\uy)$ is a cover, $T\models \psi(\uy)\to (\neg \ell_1(\uy,\uz)\lor \dots \lor \neg \ell_m(\uy,\uz))$, which implies that $\cM \models \neg \ell_j(\ua,\ub)$ for some $j=1,\dots,m$, which is a contradiction. Thus, $\psi(\uy)$ satisfies conditions (ii) too. 
\end{proof}

\subsection{Proof of Lemma~\ref{lem:si}}

\vskip 2mm\noindent
\textbf{Lemma \ref{lem:si}} \emph{
 If $T$ is stably infinite, then every finite or countable model $\cM$ of $T$ can be embedded in a model $\cN$ of $T$ such
 that  $\vert \cN\vert \setminus \vert \cM\vert$ is countable.
}

\begin{proof}
 Consider $T\cup\Delta(\cM)\cup \{c_i\neq a\, \mid\, a\in \vert \cM\vert\}_i\cup \{c_i\neq c_j\}_{i\neq j}$, where $\{c_i\}_i$ is a countable set of fresh constants: by the Diagram Lemma and the downward L\"owenheim-Skolem theorem~\cite{CK}, it is sufficient to show that this set is consistent. Suppose not; then by compactness $T\cup\Delta_0\cup \Delta_1\cup\{c_i\neq c_j\}_{i\neq j}$ is not satisfiable, for a finite subset $\Delta_0$ of $\Delta(\cM)$ and a finite subset $\Delta_1$ of 
 $ \{c_i\neq a\, \mid\, a\in \vert \cM\vert\}_i$. However, this is a contradiction because by stable infiniteness $\Delta_0$ (being satisfiable in $\cM$) is satisfiable in an infinite model of $T$. 
\end{proof}

\subsection{Proof of Lemma~\ref{lem1}}

In order to prove Lemma~\ref{lem1}, we need further background from~\cite{bgr-acmtocl} concerning amalgamation and strong amalgamation. 

\begin{definition}
  \label{def:amalgamation}
  A universal theory $T$ has the \emph{amalgamation property} iff whenever
  we are given models $\cM_1$ and $\cM_2$ of $T$ and a common
  substructure $\cM_0$ of them, there exists a further model $\cM$ of
  $T$ endowed with embeddings $\mu_1:\cM_1 \longrightarrow \cM$ and
  $\mu_2:\cM_2 \longrightarrow \cM$ whose restrictions to $|\cM_0|$
  coincide.
    %

  A universal theory $T$ has the \emph{strong amalgamation property} if the
  above embeddings $\mu_1, \mu_2$ and the above model $\cM$ can be chosen so as to satisfy the following additional
  condition: if for some $m_1, m_2$ we have $\mu_1(m_1)=\mu_2(m_2)$,
  then there exists an element $a$ in $|\cM_0|$ such that $m_1=a=m_2$.
\end{definition}

Amalgamation and strong amalgamation are strictly related to quantifier-free interpolation and to combined quantifier-free interpolation, as the result below show:


\begin{theorem}\cite{bgr-acmtocl}
  \label{thm:strong_amalgamation_syntactic} The following two conditions are equivalent for a 
   convex universal theory $T$: (i) $T$ is equality interpolating and has quan\-ti\-fier-free interpolation; (ii) $T$
  has the strong amalgamation property.
\end{theorem}

\begin{proof}
 For the sake of completeness, we report the proof of the implication (i) $\Rightarrow$ (ii) (this is the only fact used in the paper). Suppose that $T$ is equality interpolating and has quantifier-free interpolation; we prove that  it is strongly amalgamable.
 If the latter property
  fails, by Robinson Diagram Lemma, there exist models $\cM_1, \cM_2$ of
  $T$ together with a shared submodel $\cA$ such that the set of
  sentences
  $$
  \Delta_{\Sigma}(\cM_1)\cup
  \Delta_{\Sigma}(\cM_2)\cup \{ m_1\neq m_2~\vert~ m_1\in \vert
  \cM_1\vert\setminus \vert \cA\vert, 
  ~ m_2\in \vert \cM_2\vert\setminus \vert \cA\vert\} 
  $$ 
  is not $T$-consistent. By compactness, the sentence
  $$
  \delta_1(\ua, \um_1)\wedge \delta_{2}(\ua, \um_2) \rightarrow
  \bigvee_{n_1\in \um_1, n_2\in \um_2} n_1= n_2
  $$
  is $T$-valid, for some tuples $\ua\subseteq\vert \cA\vert$,
  $\um_1\subseteq(\vert \cM_1\vert\setminus \vert \cA\vert)$,
  $\um_2\subseteq(\vert \cM_2\vert\setminus \vert \cA\vert)$ and for
  some ground formulae $\delta_1(\ua, \um_1),\delta_{2}(\ua, \um_2)$
  true in $\cM_1, \cM_2$, respectively.
  %
  If the disjunction is empty, we get $T\models \delta_1(\ua, \um_1) \to \neg \delta_{2}(\ua, \um_2)$ and then we get a contradiction by the quantifier-free interpolation property (the argument is the same as below).
  Otherwise, by convexity, there are $n_1\in \um_1, n_2\in \um_2$ such that
  $$
  \delta_1(\ua, \um_1)\wedge \delta_{2}(\ua, \um_2) \rightarrow
  n_1=n_2 
  $$
   is $T$-valid.
  By the equality interpolating property, there is a term $t(\ua)$ such that
  $$
  \delta_1(\ua, \um_1)\wedge \delta_{2}(\ua, \um_2) \rightarrow
  n_1=t(\ua) 
  $$
   is $T$-valid.
   By the quantifier-free interpolation property, there is a quantifier-free formula $\theta(\ua)$ such that
   $$
  \delta_1(\ua, \um_1)\wedge n_1\neq t(\ua) \rightarrow
  \theta(\ua) 
  $$
  and 
  $$
  \theta(\ua) \to \neg
  \delta_{2}(\ua, \um_2)
  $$
   are both $T$-valid. Since $n_1\in \vert \cM_1\vert\setminus \vert \cA\vert$, we have that 
   $n_1\neq t(\ua)$ is true in $\cM_1$.
 But then we have a contradiction  because $\theta(\ua)$ is true in $\cM_1$, $\cA$ and in $\cM_2$ as well (truth of quantifier-free formulae  moves back and forth via substructures).
\end{proof}

We underline that Theorem~\ref{thm:strong_amalgamation_syntactic} extends also to the non convex case provided the notion of an equality interpolating theory is suitably adjusted~\cite{bgr-acmtocl}.

Let us now come back to the proof of Lemma~\ref{lem1}.
For proving it, we fixed a convex, stably infinite, equality interpolating, universal theory $T$ admitting a model completion $T^*$ in a signature
$\Sigma$. We fixed also 
\emph{a  $\Sigma$-constraint $\phi(\ux, \uy)$}, 
where we assumed that $\uy=y_1, \dots, y_n$.

Since $T$ has a model completion, it has uniform quantifier-free interpolants~by Theorem\ref{prop:qe}, hence  it has also (ordinary) quantifier-free interpolants. 
By Theorem~\ref{thm:strong_amalgamation_syntactic} it is  strongly amalgamable  because it is equality interpolating. In conclusion, \emph{we are allowed to use strong amalgamation in the proof below}.

Recall that for $i=1, \dots, n$, the formula $\mathtt{ImplDef}_{\phi,y_i}^T(\ux)$ was defined as the quantifier-free formula equivalent in $T^*$ to the formula
\begin{equation*}
 \forall \uy\, \forall \uy' (\phi(\ux, \uy) \wedge \phi(\ux, \uy')\to y_i= y'_i)
\end{equation*}
where the $\uy'$ are renamed copies of the $\uy$. 

\vskip 2mm\noindent
\textbf{Lemma}~\ref{lem1} 
\emph{
 Suppose that we are given a  model $\cM$ of $T$ and elements $\ua$ from the support of $\cM$ such that  $\cM\not\models \mathtt{ImplDef}_{\phi,y_i}^T(\ua)$ for all $i=1, \dots,n$. Then there exists an extension $\cN$ of $\cM$ such that 
 for some $\ub\in \vert\cN\vert \setminus \vert \cM\vert$ we have $\cN\models \phi(\ua, \ub)$. 
}

\begin{proof}
 By strong amalgamability, we can freely assume that $\cM$ is generated, as a $\Sigma$-structure, by the $\ua$: in fact, if we prove the statement for the substructure generated by the $\ua$, then strong amalgamability will provide the model we want.
 
 By using the Robinson Diagram Lemma, what we need is to prove the consistency of $T\cup \Delta(\cM)$ with the set of ground sentences  
 $$
 \{\phi(\ua, \ub)\} \cup \{ b_i\neq t(\ua)\}_{t,b_i}
 $$
 where $t(\ux)$ varies over $\Sigma(\ux)$-terms,  the $\ub=b_1, \dots, b_n$ are fresh constants  and $i$ vary over $1, \dots, n$.
 By convexity,\footnote{
 Strictly speaking, convexity says that if, for a set of literals $\phi$ and for a non empty disjunction of \emph{variables}  $\bigvee_{i=1}^n x_i=y_i$, we have $T\models \phi \to \bigvee_{i=1}^n x_i=y_i$, then we have also
 $T\models \phi \to x_i=y_i$ for some $i=1, \dots,n$. If, instead of variables, we have \emph{terms}, the same property nevertheless applies: if we have  $T\models \phi \to \bigvee_{i=1}^n t_i=u_i$, then for fresh variables $x_i, y_i$ we get $T\models \phi \wedge \bigwedge_{i=1}^n (x_i= t_i \wedge y_i=u_i)\to \bigvee_{i=1}^n x_i=y_i$, etc.
 } this set is inconsistent iff there exist a term $t(\ux)$ and $i=1, \dots, n$ such that
 $$
 T\cup \Delta(\cM)\vdash \phi(\ua, \uy)\to y_i=t~~.
 $$
 This however implies that $T\cup \Delta(\cM)$ has the formula
 $$\forall \uy\, \forall \uy' (\phi(\ua, \uy) \wedge \phi(\ua, \uy')\to y_i= y'_i)$$
 as a logical consequence.
 If we now embed $\cM$ into a model $\cN$ of $T^*$, we have that $\cN\models\mathtt{ImplDef}_{\phi,y_i}^T(\ua)$, which is in contrast to $\cM\not\models \mathtt{ImplDef}_{\phi,y_i}^T(\ua)$ (because $\cM$ is a substructure of $\cN$ and  $\mathtt{ImplDef}_{\phi,y_i}^T(\ua)$ is quantifier-free).
\end{proof}

The following Lemma will be useful in the next Subsection:

\begin{lemma}\label{lem:app} Let $T$ have a model completion $T^*$ and let the constraint $\phi(\ux,\uy)$ be of the kind 
$\alpha(\ux)\wedge \phi'(\ux,\uy)$, where $\uy=y_1,\dots, y_n$. 
 Then for every $i=1, \dots, n$, the formula $\mathtt{ImplDef}_{\phi,y_i}^T(\ux)$ is $T$-equivalent to $\alpha(\ux)\to \mathtt{ImplDef}_{\phi,y_i}^T(\ux)$.
\end{lemma}

\begin{proof}
 According to~\eqref{eq:def}, the formula  $\mathtt{ImplDef}_{\phi,y_i}^T(\ux)$ is obtained by eliminating quantifiers in $T^*$ from 
 \begin{equation}\label{eq:lemapp}
   \forall \uy\, \forall \uy' (\alpha(\ux)\wedge \phi'(\ux, \uy) \wedge \alpha(\ux)\wedge \phi'(\ux, \uy')\to y_i= y'_i)
 \end{equation}
The latter is equivalent, modulo logical manipulations, to 
\begin{equation}\label{eq:lemapp1}
   \alpha(\ux)\to\forall \uy\, \forall \uy' (\phi'(\ux, \uy) \wedge \phi(\ux, \uy')\to y_i= y'_i)
 \end{equation}
 whence the claim (eliminating quantifiers in $T^*$ from~\eqref{eq:lemapp} and~\eqref{eq:lemapp1}
gives  quantifiers-free $T^*$-equivalent formulae, hence also $T$-equivalent formulae because $T$ and $T^*$ prove the same quantifier-free formulae).
 \end{proof}

\subsection{Detailed Proof of Lemma~\ref{lem:toterminal}}
 
 \vskip 2mm\noindent
 \textbf{Lemma}~\ref{lem:toterminal} 
\emph{
 Every working formula is equivalent (modulo $T_1\cup T_2$) to a disjunction of terminal working formulae.
}

\begin{proof}
 To compute the required terminal working formulae, it is sufficient to apply the following non-deterministic procedure (the output is  the disjunction of all possible outcomes). 
The non-deterministic procedure applies one of the following alternatives.
\begin{description}
 \item[{\rm (1)}] Update $\psi_1$ by adding to it a disjunct from the DNF of 
 $\bigwedge_{e_i\in \ue} \neg 
 \mathtt{ImplDef}_{\psi_1,e_i}^{T_1}(\ux,\uz)$
 and $\psi_2$ by adding to it a disjunct from the DNF of 
 $\bigwedge_{e_i\in \ue} \neg \mathtt{ImplDef}_{\psi_2,e_i}^{T_2}(\ux,\uz)$;
 \item[{\rm (2.i)}] Select $e_i\in \ue$ and $h\in\{1,2\}$; then update $\psi_h$ by adding to it a disjunct $L_{ij}$ from the DNF 
 of $\mathtt{ImplDef}_{\psi_h,e_i}^{T_h}(\ux,\uz)$; the equality $e_i= t_{ij}$ (where $t_{ij}$ is the term mentioned in Lemma~\ref{lem2})\footnote{ 
 Lemma~\ref{lem2} is used taking as $\uy$ the tuple $\ue$, as $\ux$ the tuple $\ux,\uz$, as  $\phi(\ux, \uy)$ the formula $\psi_h(\ux, \uz,\ue)$ and as $\psi$ the formula $\psi_{3-h}$.
 }
 is added to $\mathtt{ExplDef}(\uz, \ux)$; the variable $e_i$ becomes in this way part of the defined variables.  
\end{description}
If alternative (1) is chosen, the procedure stops, otherwise it is recursively applied again and again: we have one truly existential variable less after applying alternative (2.i), so the procedure terminates, since eventually either no truly existential variable remains or alternative (1) is applied. The correctness of the procedure is due to the fact that  
the following formula is trivially a tautology: 
 \begin{eqnarray*}
 & 
 \left (\bigwedge_{e_i\in \ue} \neg \mathtt{ImplDef}_{\psi_1,e_i}^{T_1}(\ux,\uz)\wedge \bigwedge_{e_i\in \ue} \neg \mathtt{ImplDef}_{\psi_2,e_i}^{T_2}(\ux,\uz) \right ) \vee
 \\ 
 &
 \vee ~\bigvee_{e_i\in \ue} \mathtt{ImplDef}_{\psi_1,e_i}^{T_1}(\ux,\uz) \vee
 \bigvee_{e_i\in \ue} \mathtt{ImplDef}_{\psi_2,e_i}^{T_2}(\ux,\uz)~~~~~~
 \end{eqnarray*}
The  first disjunct is used in alternative (1), the other disjuncts in alternative (2.i). 
At the end of the procedure, we get a terminal working formula. Indeed, if no truly existential variable remains, then the working formula is trivially terminal.
It remains to prove that the working formula obtained after applying alternative (1) is indeed terminal.
Let $\psi'_k$ (for $k=1,2$) be the formula obtained from $\psi_k$ after applying alternative (1). We have that $\psi'_k$ is $\alpha(\ux,\uz)\wedge \psi_k(\ux,\uz,\ue)$, where $\alpha$ is a disjunct of the DNF of $\bigwedge_{e_i\in \ue} \neg 
 \mathtt{ImplDef}_{\psi_k,e_i}^{T_k}(\ux,\uz)$. We need to show that 
 $T_k\vdash \psi'_k\to \neg\mathtt{ImplDef}_{\psi'_k,e_j}^{T_k}(\ux,\uz)$  for every $j$. Fix such a $j$; according to Lemma~\ref{lem:app}, we must show that 
 $$
 T_k\vdash \alpha(\ux, \uz)\wedge \psi_k(\ux,\uz, \ue)\to \neg(\alpha(\ux, \uz)\to\mathtt{ImplDef}_{\psi_k,e_j}^{T_k}(\ux,\uz))
 $$
 which is indeed the case because $\alpha(\ux,\uz)$ logically implies $\neg\mathtt{ImplDef}_{\psi'_k,e_j}^{T_k}(\ux,\uz)$ being a disjunct of the DNF of  $\bigwedge_{e_i\in \ue} \neg 
 \mathtt{ImplDef}_{\psi_k,e_i}^{T_k}(\ux,\uz)$.
\end{proof}

\subsection{Proof of Lemma~\ref{lemma6}}\label{app:lemma6}

 \vskip 2mm\noindent
 \textbf{Lemma}~\ref{lemma6} 
\emph{
 Let $T$ be stably infinite with respect to a subset $\cS$ of the set of sorts of the signature of $T$. Let $\cM$ be a model of $T$ and  let, for every $S\in \cS$,  $X_S$ be a 
 superset of  $S^\cM$.
 Then there is an extension $\cN$ of $\cM$  such that for all $S\in \cS$ we have $S^\cN\supseteq X_S$. 
}
\begin{proof}
 Let us expand the signature of $T$ with the set $C$ of fresh constants (we take one constant for every $c\in X_S\setminus S^\cM$). We need to prove the $T$-consistency of $\Delta(\cM)$ 
 with a the set $D$ of disequalities asserting that all $c\in C$ are different from each other and from the names of the elements of the support of $\cM$. 
 By compactness, it is sufficient to ensure the $T$-consistency of $\Delta_0 \cup D_0$, where  $\Delta_0$ and $ D_0$ are finite subsets of $ \Delta(\cM)$ and $D$, respectively. 
 Since $\cM\models \Delta_0$, this set is $T$-consistent 
 and hence it is satisfied in a  $T$-model $\cM'$ 
 where all the sorts in $\cS$ are interpreted as infinite sets; in such $\cM'$, 
 it is trivially seen that we can interpret also the constants occurring in $D_0$ so as to make $D_0$ true too.  
\end{proof}

\end{document}